\documentclass[11pt,a4paper]{amsart}

\usepackage[dvipsnames]{xcolor}
\definecolor{webgreen}{rgb}{0,.5,0}
\definecolor{webbrown}{rgb}{.6,0,0}
\definecolor{RoyalBlue}{cmyk}{1, 0.50, 0, 0}
\usepackage[colorlinks=true, breaklinks=true, urlcolor=webbrown, linkcolor=RoyalBlue, citecolor=webgreen,backref=page]{hyperref}

\usepackage{lipsum}
\makeatletter
\g@addto@macro{\endabstract}{\@setabstract}
\newcommand{\authorfootnotes}{\renewcommand\thefootnote{\@fnsymbol\c@footnote}}%
\makeatother

\usepackage{eucal}
\usepackage{tikz}
\usepackage{pict2e}
\usepackage{todonotes}
\usetikzlibrary{decorations.markings}

\usepackage{array}

\usetikzlibrary{calc,patterns,angles,quotes}
\usetikzlibrary{arrows,calc,chains, positioning, shapes.geometric,shapes.symbols,decorations.markings,arrows.meta}

\usepackage{epsfig, graphicx, subfigure}
\usepackage{verbatim, setspace}
\usepackage{amsmath, amssymb}

\usepackage{mathptmx}
\usepackage{newtxtext,newtxmath}

\newcommand{\T} {\mathbb{T}}

\newcommand{\C} {\mathbb{C}}
\newcommand{\N} {\mathbb{N}}
\newcommand{\Z} {\mathbb{Z}}
\newcommand{\ZZ} {\mathbb{Z}_{\geq0}}

\newcommand{\al}{\alpha}

\newcommand{\Ga}{\Gamma}

\newcommand{\ep}{\varepsilon}
\newcommand{\de}{\delta}
\newcommand{\ka}{\kappa}

\newcommand{\Om}{\Omega}
\newcommand{\ze}{\zeta}

\newcommand{\di}{\displaystyle}

\newcommand{\re}{\mathrm{Re}}

\renewcommand{\det}{\mathrm{det}}

\newcommand{\ee}{\end{equation}}

\newcommand{\bq}{\begin{eqnarray}}
\newcommand{\eq}{\end{eqnarray}}

\newcommand{\ba}{\begin{array}}
\newcommand{\ea}{\end{array}}

\newcommand{\iv}{^{-1}}

\newcommand{\eps}{{\varepsilon}}

\newcommand{\qandq}{\quad \text{and} \quad}

\newcommand{\dd}{\mathrm{d}}
\newcommand{\ic}{\mathrm{i}}

\newtheorem{theorem}{Theorem}[section]
\newtheorem{corollary}[theorem]{Corollary}
\newtheorem{proposition}[theorem]{Proposition}
\newtheorem{lemma}[theorem]{Lemma}
\newtheorem{remark}[theorem]{Remark}

\newcommand{\hide}[1]{}

\numberwithin{equation}{section}
\usepackage{caption}
\begin{document}

\tikzset{->-/.style={decoration={
markings,
mark=at position #1 with {\arrow{latex}}},postaction={decorate}}}

\tikzset{-<-/.style={decoration={
markings,
mark=at position #1 with {\arrowreversed{latex}}},postaction={decorate}}}


\title{Asymptotics of bordered Toeplitz determinants and Next-to-Diagonal Ising Correlations}

\maketitle

\begin{center}
    \authorfootnotes
  Estelle Basor\footnote{American Institute of Mathematics, San Jose, CA 95112, USA, E-mail: ebasor@aimath.org}, Torsten Ehrhardt\footnote{Mathematics Department,  University of California, Santa Cruz, CA 95064, USA, E-mail:
tehrhard@ucsc.edu},
  Roozbeh Gharakhloo\footnote{Department of Mathematics, Colorado State University, Fort Collins, CO 80521, USA, E-mail: roozbeh.gharakhloo@colostate.edu}, \\ Alexander Its\footnote{Department of Mathematical Sciences, Indiana University-Purdue University Indianapolis, 402 N. Blackford St., Indianapolis, IN 46202, Blackford St., Indianapolis, IN 46202, USA. e-mail: aits@iupui.edu; St. Petersburg State University, Universitetskaya emb. 7/9, 199034, St. Petersburg,
Russia.} and
  Yuqi Li\footnote{East China Normal University, Shanghai, China, E-mail: yqli@sei.ecnu.edu.cn} \par \bigskip
\end{center}

\date{\today}

\subjclass{}


\begin{abstract}
We prove the analogue of the strong Szeg{\H o} limit theorem for a large class of bordered Toeplitz determinants. In particular, by applying our results to the formula of Au-Yang and Perk \cite{YP} for the next-to-diagonal correlations $\langle \sigma_{0,0}\sigma_{N-1,N} \rangle$ in the square lattice Ising model, we rigorously justify that the next-to-diagonal long-range order is the same as the diagonal and horizontal ones in the low temperature regime. The anisotropy-dependence of the subleading term in the asymptotics of the next-to-diagonal correlations is also established. We use Riemann-Hilbert and operator theory techniques, independently and in parallel, to prove these results.
\end{abstract}


\section{Introduction}

Starting from the seminal works of Szeg\H{o}, Kaufman and Onsager, Toeplitz determinants 
have played a very important role in many areas of analysis and mathematical physics.
Indeed, an extraordinary variety of problems in mathematics, physics, and engineering can be expressed
in terms of Toeplitz matrices and determinants. We refer to the monograph \cite{BS} and the more recent survey paper  \cite{DIK1}
for the details of the theory and applications of Toeplitz determinants. 

A growing  interest has recently developed
in the study of certain generalizations of Toeplitz determinants. Among those are the determinants
of Toeplitz + Hankel matrices - see \cite{DIK}, \cite{BE}, \cite{GI},  integrable Fredholm determinants \cite{IIKS}, \cite{D}, $2j-k$ and $j-2k$ determinants \cite{GW}. 
These determinants appear in the study of the Ising model in the zig-zag layered half-plane \cite{Chelkak}, 
in the spectral analysis of the Hankel matrices, in the theory of exactly solvable quantum models \cite{FA}, and in asymptotic analysis of moments of derivatives of characteristic polynomials $\Lambda_A(s)=\det(I-A s)$, where $A \in USp(2N), SO(2N), O^{-}(2N)$ \cite{Altug}.

In this paper we are concerned with yet another deformation of Toeplitz determinants - the so called  "bordered Toeplitz determinants". 
The latter also arise in applications, for example, in the the next-to-diagonal correlation functions for the Ising model. The goal of this paper is to launch a new research project devoted to the asymptotics of these determinants and to discuss this application in particular.

Let $\phi$ and $\psi$ be  the $L^1$-functions  on the (positively oriented)  unit
circle, 
$$
\T = \{z\in\C: |z| = 1\}.
$$
The {\it bordered Toeplitz determinant,}  $D^{B}_{N}[\phi; \psi]$, is defined as
\begin{equation}\label{btd}
D^{B}_{N}[\phi; \psi] := \det \begin{pmatrix}
	\phi_0& \cdots & \phi_{N-2} & \psi_{N-1}\\
	\phi_{-1}& \cdots  & \phi_{N-3}&\psi_{N-2}  \\
	\vdots & \vdots & \vdots & \vdots\\
	\phi_{1-N} &  \cdots  & \phi_{-1}&\psi_{0}
	\end{pmatrix} , \qquad N > 1,
\end{equation}
where 
\begin{equation}\label{phi_n} 
\phi_n = \int_{\T}z^{-n}\phi(z)\frac{\dd z}{2\pi \ic z},
\qquad
\psi_n = \int_{\T}z^{-n}\psi(z)\frac{\dd z}{2\pi \ic z},
\end{equation}
are respectively the $n$-th Fourier coefficients of $\phi$ and $\psi$.  To fix the notation, we let
 \begin{equation}\label{ToeplitzDet}
    D_N[\phi] := \underset{0 \leq j,k\leq N-1}{\det} \{ \phi_{j-k} \},
\end{equation}
     denote the $N\times N$ (pure) Toeplitz determinant corresponding to the symbol $\phi$. As with the Toeplitz determinants, the principal analytic question is the asymptotic behavior
of $D^{B}_{N}[\phi; \psi]$ as $N \to \infty.$

The asymptotics of the Toeplitz determinants are well known and given by the Szeg\H{o}-Widom theorem \cite{WIDOMBlock,Sz,BS}
\begin{equation}\label{Szego Theorem}
	D_{N}[\phi] \sim G[\phi]^{N} E[\phi],   \qquad N\to\infty,
\end{equation}
where
\begin{equation}\label{G and E}
	G[\phi] = \exp \left([\log \phi]_{0}\right) \qandq E[\phi] = \exp \left( \sum_{n \geq 1} n[\log \phi ]_{n}[\log \phi]_{-n} \right).
\end{equation}
This holds if the function $\phi$ is sufficiently smooth (e.g., in H\"older class $C^{1+\epsilon}$), does not vanish on $\T$, and has zero winding number. Note that the constants involve the $n$-th Fourier coefficients $[\log \phi]_{n}$ of the continuous logarithm
of the function $\phi$.

In this paper, we will show that for the bordered determinants a similar theorem holds,
\begin{equation}\label{maingr}
 D^{B}_{N}[\phi; \psi]  \sim G[\phi]^{N} E[\phi]\,F[\phi; \psi], \qquad N\to\infty,
 \end{equation}
and where $F[\phi; \psi]$ is a constant described in Theorems \ref{main thm} and \ref{thm 1.2} below. 

Our general result above will become explicit in the case of the next-to-diagonal Ising correlations.  There it happens that $\psi$ is a constant times something of the form 
 \begin{equation}\label{phiqpsi}
 \psi(z) =  \frac{\phi(z)z - d}{z - c},
 \end{equation}
where $d$ and $c$ \color{black} are complex parameters. Hence particular attention will be paid to such functions. In Theorem \ref{main thm} below, we present the asymptotics of  $D^{B}_{N}[\phi; \psi]$, where $\psi$ is of the more general form
\begin{equation}\label{general psi}
\psi(z) = q_1(z) \phi(z) + q_2(z),
\end{equation}
where
\begin{equation}\label{q1 q2}
q_1(z) = a_0+a_1z+\frac{b_0}{z}+\sum^{m}_{j=1}\frac{b_j z}{z-c_j}, \qandq q_2(z) = \hat{a}_0+\hat{a}_1z+\frac{\hat{b}_0}{z}+
  \sum_{j=1}^m \frac{\hat{b}_j}{z-c_j},
\end{equation}
where all parameters are complex and the $c_j$ are nonzero and do not lie on the unit circle.  Indeed it is straightforward to pass from the rational functions $q_1$ and $q_2$ with one simple pole to the ones with multiple simple poles, as one can use the following elementary properties of bordered Toeplitz determinants: \begin{equation}\label{linear-combination}
D^B_N\left[\phi; \sum^m_{j=1} a_j \psi_j\right] =\sum^m_{j=1} a_j D^B_N[\phi,\psi_j], 
\end{equation}
\begin{equation}\label{border,phi,phi}
D^B_{N}[\phi; \phi] = D_{N}[\phi],
\end{equation} 
\begin{equation}\label{border,phi,const}
D^B_{N}[\phi;1] = D_{N-1}[\phi].
\end{equation}

Throughout the the paper, we will refer to a symbol $\phi$ as a \textit{Szeg{\H o}-type} symbol, if it is smooth and nonzero on the unit circle, has winding number zero, and admits an analytic continuation in a neighborhood of the unit circle.

\begin{theorem}\label{main thm}
	Let $D^{B}_{N}[\phi; \psi]$ be the bordered Toeplitz determinant with $\psi=q_1 \phi + q_2$ given by \eqref{general psi} and \eqref{q1 q2}, and $\phi$ of Szeg{\H o} type. Then, the following asymptotic behavior of $D^{B}_{N}[\phi; \psi]$ as $N \to \infty$ takes place 
\begin{equation}\label{main formula}
	  D^B_{N}\left[\phi ; \psi \right]   =  
	G[\phi]^{N} E[\phi]\left(F[\phi;\psi] + O(e^{-\mathfrak{c}N})\right),
	\end{equation} 
	where $G[\phi]$ and $E[\phi]$ are given by \eqref{G and E},
\begin{equation}\label{ConstantF2}
	F[\phi;\psi]  =	
	a_0+ b_0 [\log \phi]_{1} + 
	\hspace{-0.15cm}\sum^m_{j=1 \atop 0<|c_j|<1} \hspace{-0.15cm} b_j \frac{\al(c_j)}{\alpha(0)} 
	+	
         \frac{1}{\al(0)}\left( 
         \hat{a}_0-\hat{a}_1[\log \phi]_{-1} -
         \hspace{-0.15cm}\sum^{m}_{j=1  \atop |c_j|>1 }\hspace{-0.15cm}\frac{\hat{b}_j}{c_j} \al(c_j) \right),
	\end{equation}
	\begin{equation}
		\al(z):= \exp \left[ \frac{1}{2 \pi i } \int_{\T} \frac{\ln(\phi(\tau))}{\tau-z}d\tau \right],
	\end{equation}
and $\mathfrak{c}$ is some positive constant.
\end{theorem}    
It is occasionally convenient to use different notation  related to the function $\alpha(z)$,
\begin{equation}\label{alpha.phi.pm}
\alpha(z)=\begin{cases} \phi_+(z), & |z|<1, \\ \phi_-^{-1}(z), & |z|>1, 
\end{cases}
\end{equation}
with
\begin{equation}\label{phi.pm}
\phi_+(z):=\exp\left(\sum_{n=0}^\infty [\log \phi]_n z^n\right),\qquad
\phi_-(z):=\exp\left(\sum_{n=1}^\infty [\log \phi]_{-n} z^{-n}\right).
\end{equation} 
In fact, these functions are the factors of a  {\em canonical Wiener-Hopf factorization} of the symbol $\phi$,
$\phi(z)=\phi_-(z)\phi_+(z)$, $|z|=1$. Factors in a Wiener-Hopf factorization are unique up to a multiplicative constant.
With the factors as given above, we have the normalization,
\begin{equation}\label{WH.norm.allpah.G}
\phi_+(0)=\alpha(0)=G[\phi],\qquad \phi_-(\infty)=1=\alpha(\infty).
\end{equation}
More generally we can find the constant $F[\phi;\psi]$ in \eqref{maingr} as described in the following theorem, which is proven using operator theory and Riemann-Hilbert methods respectively in Sections \ref{Section OT 1} and \ref{RHP proof thm 1.2}.
\begin{theorem}\label{thm 1.2}
	Let  $\psi(z)$ be a function which admits an analytic continuation in a neighborhood of the unit circle, and let 
	$\phi$ be of Szeg{\H{o}} type. Denote  by $\phi_{\pm}(z)$  the factors of a canonical Wiener-Hopf factorization of the symbol 
	$\phi(z)$, i.e.,  $\phi=\phi_-\phi_+$. Then
	\begin{equation}\label{main2}
		 D^B_{N}\left[\phi ; \psi \right]   =  
		G[\phi]^{N} E[\phi]\left( F[\phi;\psi]  + O(e^{-\mathfrak{c}N})\right),
	\end{equation}
	where $G[\phi]$ and $E[\phi]$ are given by \eqref{G and E},
	\begin{equation}\label{ConstantFthm1.2}
		F[\phi;\psi]=\frac{[\phi_-^{-1}\psi]_0}{[\phi_+]_0},
	\end{equation}
	and $\mathfrak{c}$ is some positive constant.
\end{theorem}
\begin{remark}\label{rem.1.3} \normalfont As it is shown in Section 3, with  the change to $o(1)$ in the error term, the asymptotics (\ref{main formula}) and (\ref{main2}) are valid for all  $\psi\in L^2(\T)$ and $\phi$ satisfying the
assumptions of the strong Szeg\H{o} theorem, i.e., $\phi(z)$ belongs to a H\"older class $C^{1 + \epsilon}$, is nonzero on the unit circle, and has zero winding number.  
\end{remark}
In this paper, we also apply our general results mentioned above to the problem of rigorous evaluation of the next-to-diagonal two-point correlation function in the Ising model. To that end, let us first recall more precisely the situation in the two-dimensional Ising model, solved by Onsager (see, e.g., \cite{McCoy-Wu}).
In this model a $2\mathcal{M}\times 2\mathcal{N}$ rectangular lattice is considered with
an associated
spin variable $\sigma_{jk}$ taking the values $1$ and $-1$
at each vertex $(j,k)$, $-{\mathcal M}\le j\le {\mathcal M}-1$,
$-{\mathcal N}\le k\le {\mathcal N}-1$.
There are $2^{4{\mathcal M}{\mathcal N}}$ possible spin configurations
$\{\sigma\}$ of the lattice (a configuration corresponds to values of all
$\sigma_{jk}$ fixed). By $J_h$ and $J_v$ we respectively denote the horizontal and vertical nearest neighbor coupling constants and with each configuration we associate its nearest-neighbor coupling energy given by
\begin{equation}
    E(\{\sigma\})=
-\sum_{j=-{\mathcal M}}^{{\mathcal M}-1}\sum_{k=-{\mathcal N}}^{{\mathcal N}-1}
\left(J_{h}\sigma_{jk}\sigma_{j\,k+1}+J_{v}\sigma_{jk}\sigma_{j+1\,k}\right),\qquad
J_{h}, J_{v}>0.
\end{equation}
The partition function at a temperature $T>0$ is equal to
\begin{equation}
    Z(T)=
\sum_{\{\sigma\}}e^{-E(\{\sigma\})/k_BT},
\end{equation}
where the sum is over all configurations and $k_B$ is the Boltzmann constant. A remarkable feature of this model
is the presence of a thermodynamic phase transition (in the limit of the
infinite lattice, $\mathcal{M},\mathcal{N}\to \infty$) at a certain temperature $T_c$  whose dependence on
 $J_{h}$, $J_v$ is described by the equation,
 \begin{equation}\label{Tcr}
 \sinh\left(\frac{2J_h}{k_BT_c}\right) \sinh\left(\frac{2J_v}{k_BT_c}\right) = 1.
\end{equation}
Define a 2-spin correlation function by the expression
\begin{equation}\label{corr}
    \langle\sigma_{0,0}\sigma_{N,M}\rangle=
\lim_{{\mathcal{M}},{\mathcal{N}}\to\infty}{\frac{1}{Z(T)}}
\sum_{\{\sigma\}}\sigma_{0,0}\sigma_{N,M}e^{-E(\{\sigma\})/k_BT}.
\end{equation}
Let us introduce the notations,
\begin{equation}\label{SSCC}
\begin{split}\begin{aligned}
S_h &= \sinh\left(\frac{2J_h}{k_BT}\right), &\;\; S_v &= \sinh\left(\frac{2J_v}{k_BT}\right)\,\, ,
\\[1ex]
C_h &= \cosh\left(\frac{2J_h}{k_BT}\right), & C_v &= \cosh\left(\frac{2J_v}{k_BT}\right)\,\,,
\end{aligned}\end{split}
\end{equation}
and 
\begin{equation}\label{k-parameter}
    k=S_hS_v.
\end{equation}
In this paper we shall focus on \begin{equation}\label{k>1}
k>1,
\end{equation}
which,  in view of equation (\ref{Tcr}), corresponds to the  low temperature regime  $T<T_c$. It is known (see, e.g., \cite[Chap.~VIII]{McCoy-Wu})  that the diagonal correlations $\langle \sigma_{0,0}\sigma_{N,N} \rangle$ and the horizontal correlations $\langle \sigma_{0,0}\sigma_{0,N} \rangle$ have Toeplitz determinant representations. Indeed, we have 
\begin{align}\label{Isingphi}
\langle \sigma_{0,0}\sigma_{N,N} \rangle & = D_N[ \widehat{\phi} \ ], \qquad \widehat{\phi}(z) = \sqrt{\frac{1-k^{-1}z^{-1}}{1-k^{-1}z}},
\\[1ex]
\label{NTD horizontal corr toeplitz}
\langle \sigma_{0,0}\sigma_{0,N} \rangle &= D_N[ \widehat{\eta} \ ],\qquad\, \widehat{\eta}(z) = \sqrt{\frac{(1-\al_1z)(1-\al_2z^{-1})}{(1-\al_1z^{-1})(1-\al_2z)}},
\end{align} 
where the constants $\al_1$ and $\al_2$ are given by
\[\al_1=\frac{z_h(1-z_v)}{1+z_v}, \quad \al_2=\frac{1-z_v}{z_h(1+z_v)}, \qquad z_{h,v} = \tanh{\left(\frac{J_{h,v}}{k_BT}\right)}.\] 
In the low temperature regime, the symbols $\widehat{\phi}$ and $\widehat{\eta}$ enjoy the regularity properties required by the strong Szeg{\H o} limit theorem and the diagonal and horizontal long-range orders 
\[ M_D := \sqrt{\lim_{N\to\infty}\langle \sigma_{0,0}\sigma_{N,N} \rangle} \qandq M_H := \sqrt{\lim_{N\to\infty}\langle \sigma_{0,0}\sigma_{0,N} \rangle}, \]
both evaluate to $(1-k^{-2})^{1/8}$ (see \cite[Chap.~XI]{McCoy-Wu}).

In an interesting development, it was shown by Au-Yang and Perk in \cite{YP}, that the next-to-diagonal two point correlation function is given by the following bordered Toeplitz determinant,
 \begin{equation}\label{BT&NTD}
 	\langle \sigma_{0,0}\sigma_{N-1,N} \rangle = D^B_N[\widehat{\phi}; \widehat{\psi}],
 \end{equation}
 where $\widehat{\phi}$ is given in \eqref{Isingphi}, and 
  \begin{equation}\label{hat psi}
 \widehat{\psi}(z)= \frac{C_v z\widehat{\phi}(z)+C_h}{S_v(z-c_*)}, \qquad \mbox{with} \qquad c_* = -\frac{S_h}{S_v}.
 \end{equation}
This is straightforward to derive these formulae from the original expressions in \cite{YP}, and we have provided it as an appendix in Section \ref{Simplifying AuYangPerk}. We would like to emphasize that in the low-temperature regime ($k>1$) and in the anisotropic case ($J_h \neq J_v$),
 the symbols $\widehat{\phi}$ and $\widehat{\psi}$ satisfy the corresponding assumptions of Theorem \ref{main thm}, in particular, $\widehat{\phi}$ is of Szeg\H{o} type. The function $\widehat{\psi}(z)$ actually does {\em not} have a pole at $z=c_*$, and therefore it is analytic on a neighborhood of
 the unit circle, even in the isotropic case when $c_*=-1$.
 
Our results being applied to the next-to-diagonal theory for the Ising model show the following large $N$ behavior of the corresponding
correlation function in the  low temperature regime ($k>1$), which is valid in both the isotropic and anisotropic cases.
 \begin{theorem}\label{th2}
Let $\langle \sigma_{0,0}\sigma_{N-1,N} \rangle$ be the next-to-diagonal two point correlation function
in the square lattice Ising model. Then,  in the low-temperature regime, the long-range order in the next-to-diagonal direction 
is the same as of the diagonal and horizontal ones,
i.e., 
\begin{equation}\label{"magnetization"}
\lim_{N \to \infty} \langle \sigma_{0,0}\sigma_{N-1,N} \rangle=  (1-k^{-2})^{1/4}.
\end{equation}
\end{theorem}
It is worth noticing that, although the bordered Toeplitz determinant which defines the correlation 
function  $\langle \sigma_{0,0}\sigma_{N-1,N} \rangle$ depends on the relation between
$J_h$ and $J_v$, its leading order asymptotics does not.  However, the sensitivity to the horizontal and vertical parameters is reflected in the second-order term of the asymptotic expansion as our next theorem illustrates.
\begin{theorem} \label{th22}
The next-to-diagonal two point correlation function has, in the low-temperature regime $k>1$, the $N\to\infty$ asymptotics
\begin{align}\label{magnetization2}
	\langle \sigma_{0,0}\sigma_{N-1,N} \rangle &=
	(1-k^{-2})^{1/4}\left(1+\frac{1}{2\pi(1-k^{-2})}\Big(\frac{1}{C_v^2}+\frac{1}{k^2-1}\Big)N^{-2}k^{-2N}\Big(1+O(N^{-1})
	\Big)
	\right).
\end{align}
\end{theorem}
For comparison, asymptotics of the diagonal correlation function is given by
\begin{align}\label{magnetization3.diag}
\langle \sigma_{0,0}\sigma_{N,N} \rangle 
&=
	(1-k^{-2})^{1/4}\left(1+\frac{1}{2\pi(1-k^{-2})^2 k^2}N^{-2}k^{-2N}\Big(1+O(N^{-1})\Big)	\right)
	,\qquad N\to\infty
\end{align}
(see formula (3.27) in Chap.~XI of \cite{McCoy-Wu}).
As part of our computation leading to \eqref{magnetization2}, we reconfirm \eqref{magnetization3.diag}
as well.

The critical temperature  $(k=1)$ and the high temperature regime
$(k<1)$ correspond to the appearance of the Fisher-Hartwig type singularities in the symbol (\ref{Isingphi}) and will be considered in a future publication. 

It also should be mentioned that  Theorems \ref{th2} and  \ref{th22} confirm the  
long-range behavior of the next-to-diagonal correlation functions of the Ising model that  have already been 
known  in the physical literature \cite{CW}{\footnote{ In fact, in \cite{CW}, the long-range asymptotics
is obtained for the general correlation function 
$ \langle \sigma_{0,0}\sigma_{M,N} \rangle $.}.}

Finally we remark that the constant $F[\phi; \psi]$ can actually vanish for certain $\psi$. This happens, for example, if $\psi = \phi\frac{z}{z - c}$ with $|c| >1$ as can be seen from \eqref{ConstantF2}). In this case the  second-order term in the asymptotics  becomes important.

In this paper we shall present two different approaches to the general problem of bordered determinants. 
One is  based on the relatively new Riemann-Hilbert method of the asymptotic analysis 
of  Toeplitz and Hankel  determinants (see \cite{BDJ}, \cite{FIK}, \cite{DIK}). The Riemann-Hilbert approach has been inspired by the work \cite{W} where the connection of bordered Toeplitz determinants of the type $D^B_N[\phi;q\phi]$ to the system of biorthogonal polynomials on the unit circle was found for the first time. Another approach is based on the operator theoretic techniques, and it has been used in the theory of Toeplitz and Hankel determinants 
since the classical works of Szeg\H{o} and Widom (see \cite{WIDOMBlock}, \cite{Sz}, \cite{BS}, \cite{BS1}, \cite{BE}, \cite{BE1}). For the last 25 years these two techniques has been very closely 
interacting and greatly enhancing  each other. In particular, the asymptotic analysis of the bordered Toeplitz determinants
whose results are presented in this work has been carried   out within constant interaction and information exchanges
between  the first two and the last three co-authors. Hence we decided that it would be very  proper   to present both
the operator and Riemann-Hilbert  methods of the solution in one paper.  

\subsection{Outline.}

The paper is organized as follows. In Section \ref{BTD RHP} we shall present the Riemann-Hilbert representation of the bordered Toeplitz determinant corresponding to a symbol pair $(\phi, \psi)$, $\psi$ given by \eqref{general psi} and \eqref{q1 q2}. In this section we shall basically follow \cite{W} where the connection with the corresponding system of bi-orthogonal polynomials on the unit circle was first obtained. We will then prove Theorems \ref{main thm}, \ref{th2}, and \ref{th22} based on the Riemann-Hilbert formulation. Theorems \ref{main thm}, \ref{thm 1.2}, \ref{th2} and \ref{th22} will be proven using operator theory techniques in Section \ref{Section OT}.
In Section \ref{Sec:Num} a numerical verification for the asymptotics of the correlation function in Theorem \ref{th22} as well as for the asymptotics
of $D_N^B[\phi;\psi]$ in the case $\psi=\phi\frac{z}{z-c}$ is done. Finally Section \ref{Appendices} contains four appendices respectively on the solution of the associated Riemann-Hilbert problem, proof of Theorem \ref{thm 1.2} using the Riemann-Hilbert approach, derivation of the Ising symbol pair $(\widehat{\phi},\widehat{\psi})$, and some other auxiliary results.


\section{Bordered Toeplitz determinants and the Riemann-Hilbert problem for bi-orthogonal polynomials on the unit circle}\label{BTD RHP}
As mentioned in the outline, the goal of this section is to prove Theorem \ref{main thm}. In order to achieve that, we will first establish the relationship between the bordered Toeplitz determinant $D^{B}_N[\phi;\psi]$, $\psi$ given by \eqref{general psi}, and the solution of the Riemann-Hilbert problem for the system of bi-orthogonal polynomials on the unit circle (BOPUC). Let $Q_n$ and $\widehat{Q}_n$ be respectively defined by
\begin{equation}\label{Toeplitz OP 1}
Q_n(z):= \frac{1}{\sqrt{D_n[\phi] D_{n+1}[\phi]}} \det \begin{pmatrix}
\phi_0 & \phi_{-1} & \cdots & \phi_{-n} \\
\phi_1 & \phi_{0} & \cdots & \phi_{-n+1} \\
\vdots & \vdots & \ddots & \vdots \\
\phi_{n-1} & \phi_{n-2} & \cdots & \phi_{-1} \\
1 & z & \cdots & z^n
\end{pmatrix},
\end{equation}
and
\begin{equation}\label{Toeplitz OP 2}
\widehat{Q}_n(z) := \frac{1}{\sqrt{D_n[\phi] D_{n+1}[\phi]}} \det \begin{pmatrix}
\phi_0 & \phi_{-1} & \cdots & \phi_{-n+1} & 1 \\
\phi_1 & \phi_{0} & \cdots & \phi_{-n+2} & z \\
\vdots & \vdots & \ddots & \vdots \\
\phi_{n} & \phi_{n-1} & \cdots  & \phi_{1} & z^n
\end{pmatrix},
\end{equation}
where $\phi_j$, $j \in \Z$, are defined by \eqref{phi_n} and $D_n[\phi]$ is given by \eqref{ToeplitzDet}. Note that, from \eqref{Toeplitz OP 1} and \eqref{Toeplitz OP 2}, we have
\begin{equation}
    Q_n(z)=\ka_nz^n + \sum^{n-1}_{\ell=0} c_{\ell} z^{\ell} , \qquad \mbox{and} \qquad \widehat{Q}_n(z)=\ka_nz^n + \sum^{n-1}_{\ell=0} \widehat{c}_{\ell} z^{\ell},
\end{equation}
where
\begin{equation}\label{ka_n}
\ka_n = \sqrt{\frac{D_n[\phi]}{D_{n+1}[\phi]}}.
\end{equation}
One can readily observe that $\{Q_n\}^{\infty}_{n=0}$ and $\{\widehat{Q}_n\}^{\infty}_{n=0}$ form the bi-orthogonal system of polynomials on the unit circle with respect to the weight $\phi$ :
\begin{equation}
\int_{\T} Q_n(z)\widehat{Q}_n(z^{-1})\phi(z) \frac{\dd z}{2\pi \ic z} = \de_{nk}, \qquad n,k=0,1,2,\cdots.
\end{equation}
It is due to J.Baik, P.Deift and K.Johansson (\cite{BDJ}) that the following matrix-valued function constructed out of the polynomials $Q_n$ and $\widehat{Q}_n$
\begin{equation}\label{Toeplitz-OP-solution}
X(z;n):=\begin{pmatrix}
\ka_n^{-1} Q_n(z) & \di \ka^{-1}_n \int_{\T} \frac{Q_n(\ze)}{(\ze-z)} \frac{\phi(\ze)\dd \ze}{2\pi \ic \ze^n} \\
-\ka_{n-1}z^{n-1}\widehat{Q}_{n-1}(z^{-1}) & \di -\ka_{n-1} \int_{\T}  \frac{\widehat{Q}_{n-1}(\ze^{-1})}{(\ze-z)} \frac{\phi(\ze)\dd \ze}{2\pi \ic \ze}
\end{pmatrix},
\end{equation}
satisfies the following Riemann-Hilbert problem for BOPUC, which in the subsequent parts of this text will occasionally be referred to as the $X$-RHP:
\begin{itemize}
\item  \textbf{RH-X1} \qquad $X:\C\setminus \T \to \C^{2\times2}$ is analytic,
\item \textbf{RH-X2} \qquad  The limits of $X(\ze)$ as $\ze$ tends to $z \in \T $ from the inside and outside of the unit circle exist, and are denoted $X_{\pm}(z)$ respectively and are related by
\begin{equation}
X_+(z)=X_-(z)\begin{pmatrix}
1 & z^{-n}\phi(z) \\
0 & 1
\end{pmatrix}, \qquad  z \in \T,
\end{equation}
\item \textbf{RH-X3} \qquad  As $z \to \infty$
\begin{equation}
X(z)=\big( I + O(z^{-1}) \big) z^{n \sigma_3}.  
\end{equation}
\end{itemize}
 For convenience of the reader, in the Appendix \ref{Appendix Y-RHP} we have provided the solution of the $X$-RHP when $\phi$ is of Szeg{\H o} type. 
 
 In the following subsections we will analyze bordered Toeplitz determinants of the following three types \begin{itemize}
    \item $D^B_{N}[\phi;z^k]$,
    \item $D^B_{N}[\phi;q],$
    \item $D^B_{N}[\phi;q\phi]$,
\end{itemize}
where $q$ is a rational function with simple poles. The lemmas in the following subsections, whose proofs are inspired by calculations in \cite{W}, show that the bordered Toeplitz determinants of the above types are encoded into the solution of the $X$-RHP. In fact, we will show that the bordered Toeplitz determinants of the first two types are related to the $X_{11}$ and the bordered Toeplitz determinants of the third type are related to the $X_{12}$, respectively, the $11$ and $12$ entries of the solution of the $X$-RHP. Later we will show how these cases are relevant to the next-to-diagonal correlations in the 2D-Ising model.

\subsection{Bordered Toeplitz determinants of the type $D^B_{N}[\phi;z^k],  k \in \Z$}
\label{monomials} Let us start this subsection with the following elementary lemma.
\begin{lemma}\label{elementary}
The following identity holds for the Bordered Toeplitz determinants
\begin{equation}\label{negative powers}
    D^B_{n+1}[\phi; z^{k}] = 0, \qquad k \in \Z \setminus \{ 0,1, \cdots, n \}.
\end{equation}
\end{lemma}
\begin{proof}
It suffices to note that all Fourier coefficients $(z^k)_j=0$ for $0\leq j \leq n$, $k \in \Z \setminus \{ 0,1, \cdots n \}$.
\end{proof}
Note that for $k=0$, we obviously have \eqref{border,phi,const}. Now, we turn our attention to \begin{equation}
    D^B_{n+1}[\phi;z^{k}], \qquad k \in \{1,\cdots,n\}.
\end{equation}From
\begin{equation}\label{boreder-z^k-positive-powers}
    D^B_{n+1}[\phi; z^{k}] = \det \begin{pmatrix}
\phi_0 & \phi_{-1} & \cdots & \phi_{k-n+1} & \phi_{k-n}  & \phi_{k-n-1} & \cdots &  \phi_{-n} \\
\phi_1 & \phi_{0} & \cdots & \phi_{k-n+2}  & \phi_{k-n+1} & \phi_{k-n} & \cdots & \phi_{-n+1} \\
\vdots & \vdots & \cdots & \vdots \\
\phi_{n-1} & \phi_{n-2} & \cdots & \phi_{k} & \phi_{k-1} & \phi_{k-2} & \cdots & \phi_{-1} \\
0 & 0 & \cdots & 0 & 1 & 0 & \cdots & 0 
\end{pmatrix}.
\end{equation}
we observe that (by \eqref{Toeplitz OP 1} and \eqref{ka_n}) the determinant on the right hand side of \eqref{boreder-z^k-positive-powers} is exactly the coefficient of $z^{n-k}$ in the polynomial \begin{equation}
    \ka^{-1}_n D_n[\phi] Q_n(z).
\end{equation}
Let \begin{equation}
    Q_n(z) \equiv \sum^{n}_{j=0} \ka^{(n)}_j z^j.
\end{equation} where for brevity of notation, throughout this paper we use \begin{equation}
    \ka_n \equiv \ka^{(n)}_n.
\end{equation} Therefore
\begin{equation}
    D^B_{n+1}[\phi;z^{k}] = D_n[\phi] 
    \frac{\ka^{(n)}_{n-k}}{\ka_n}.
\end{equation}
We are now in a position to express $D^B_{n+1}[\phi;z^{k}]$, $1 \leq k \leq n$, in terms of $X$-RHP data in a recursive way as follows:
\begin{equation}\label{border-z-X-RHP}
    D^B_{n+1}[\phi;z] = D_n[\phi] \lim_{z \to \infty} \left( \frac{X_{11}(z;n)-z^n}{z^{n-1}} \right) \equiv D_n[\phi] \frac{\ka^{(n)}_{n-1}}{\ka_n},
\end{equation}
\begin{equation}\label{border-z^2-X-RHP}
    D^B_{n+1}[\phi;z^2] = D_n[\phi] \lim_{z \to \infty} \left( \frac{X_{11}(z;n)-z^n - \frac{\ka^{(n)}_{n-1}}{\ka_n}z^{n-1} }{z^{n-2}} \right) \equiv D_n[\phi] \frac{\ka^{(n)}_{n-2}}{\ka_n},
\end{equation}
and so on. These formulae are recursive, in the sense that the second and third members of the equality \eqref{border-z-X-RHP} can be regarded as the definition of $k^{(n)}_{n-1}$ in terms of the $X$-RHP, which one needs in \eqref{border-z^2-X-RHP}.

Here, in particular we present how the asymptotics of $D^B_{n+1}[\phi;z]$ can be obtained from the Riemann-Hilbert data. In lemma (\ref{analogue of lemma 2.2 for c=0}) we will show that this is actually related to $D_n[\phi;\frac{\phi}{z}]$.
\begin{lemma}\label{psi=z}
Let $\phi$ be of Szeg{\H o} type. Then, as $n \to \infty$ we have
\begin{equation}\label{phi,z asymp1}
    D^B_{n+1}[\phi;z] =  D_n[\phi] \left(-\frac{1}{2\pi \ic} \int_{\T} \ln(\phi(\tau)) \dd \tau + O(e^{-\mathfrak{c}n})\right), 
\end{equation}
for some positive constant $\mathfrak{c}$.
\end{lemma}
\begin{proof}
    Expanding $\al(z)$, given by \eqref{45}, as $z \to \infty $ we get 
\begin{equation}\label{expansion of al(z) at infinity}
    \al(z) = 1 - \frac{a_0}{2\pi \ic z} - \left( \frac{a_1}{2\pi \ic} + \frac{a^2_0}{8 \pi^2} \right) \frac{1}{z^2} + \cdots , \qquad z \to \infty,
\end{equation}
where
\begin{equation}
    a_k := \int_{\T} \tau^k \ln \left(\phi(\tau)\right) \dd \tau. 
\end{equation}
Also from \eqref{X in terms of R exact}, and \eqref{R asymp} we have

\begin{equation}\label{X11}
    X_{11}(z;n) = \al(z)z^n\left(1 + \frac{ O(e^{-2\mathfrak{c}n})}{ 1+|z| } \right), \qquad z \in \Om_{\infty}, \qquad n \to \infty.
\end{equation}
Combining \eqref{border-z-X-RHP}, \eqref{expansion of al(z) at infinity} and \eqref{X11} gives \eqref{phi,z asymp1}.
\end{proof}
In a similar fashion, and with increasing effort, one can obtain similar formulae for $D^B_{n+1}[\phi;z^k]$, $k>1$. 


 \subsection{Bordered Toeplitz determinants of the type $D^B_{N}[\phi;q]$}\label{rationals}  
 Let us define
\begin{equation}\label{q1 and q2}
    q_0(z):= \frac{1}{z-c}.
\end{equation}
The Fourier coefficients of $q_0$ are given by
\begin{equation}\label{q_kj}
    q_{0,j} = \begin{cases}
    0, & |c|<1, \\
    -(c)^{-j-1}, & |c|>1,
    \end{cases}\qquad 0\leq j \leq n.
\end{equation}
The following lemma establishes how $D^B_{N}[\phi;q_0]$ is encoded into $X$-RHP data.
\begin{lemma}
The bordered Toeplitz determinant $D^B_{n+1}[\phi,\di \frac{1}{z-c}]$, is encoded into $X$-RHP data described by
\begin{equation}\label{psi=1/z-c}
    D^B_{n+1}[\phi;\frac{1}{z-c}] = \begin{cases}
    0, & |c|<1, \\
    -c^{-n-1}D_{n}[\phi]X_{11}(c;n), & |c|>1,
    \end{cases}
\end{equation}
where $D_{n}[\phi]$ is given by \eqref{ToeplitzDet} and $X_{11}$ is the $11$ entry of the solution to  \textbf{RH-X1} through  \textbf{RH-X3}.
\end{lemma}
\begin{proof}
The case of $|c|<1$ is obvious due to \eqref{q_kj}. Consider $|c|>1$. Recalling that $X_{11}(z;n)=\ka^{-1}_nQ_n(z)$, from \eqref{Toeplitz OP 1} and \eqref{ka_n} we have
  \begin{equation}
        X_{11}(z;n) = \frac{1}{D_n[\phi]} \det \begin{pmatrix}
\phi_0 & \phi_{-1} & \cdots & \phi_{-n} \\
\phi_1 & \phi_{0} & \cdots & \phi_{-n+1} \\
\vdots & \vdots & \cdots & \vdots \\
\phi_{n-1} & \phi_{n-2} & \cdots & \phi_{-1} \\
1 & z & \cdots & z^{n}
\end{pmatrix}.
\end{equation}
Therefore from \eqref{q_kj}
\begin{equation}
    \begin{split}
        -c^{-n-1}D_{n}[\phi]X_{11}(c;n) & = \det \begin{pmatrix}
\phi_0 & \phi_{-1} & \cdots & \phi_{-n} \\
\phi_1 & \phi_{0} & \cdots & \phi_{-n+1} \\
\vdots & \vdots & \cdots & \vdots \\
\phi_{n-1} & \phi_{n-2} & \cdots & \phi_{-1} \\
-c^{-n-1} & -c^{-n} & \cdots & -c
\end{pmatrix} \\ & = \det \begin{pmatrix}
\phi_0 & \phi_{-1} & \cdots & \phi_{-n} \\
\phi_1 & \phi_{0} & \cdots & \phi_{-n+1} \\
\vdots & \vdots & \cdots & \vdots \\
\phi_{n-1} & \phi_{n-2} & \cdots & \phi_{-1} \\
q_{0,n} & q_{0,n-1} & \cdots & q_{0,0}
\end{pmatrix} \equiv D^B_{n+1}[\phi,q_0].
    \end{split}
\end{equation}
\end{proof}
\begin{corollary}\label{cor 1}
We have
\begin{equation}\label{228}
	D^B_{n+1}\left[\phi;\di a+\frac{b_0}{z}+\sum^{m}_{j=1}\frac{b_j }{z-c_j}\right] = D_n[\phi] \left(a -\sum^{m}_{j=1  \atop |c_j|>1 }b_j c^{-n-1}_jX_{11}(c_j;n)\right),
\end{equation}
and for a Szeg{\H o} type $\phi$
\begin{equation}\label{229}
D^B_{n+1}\left[\phi;\di a+\frac{b_0}{z}+\sum^{m}_{j=1}\frac{b_j }{z-c_j}\right] = G[\phi]^{n} E[\phi] \left(a -\sum^{m}_{j=1  \atop |c_j|>1 }\frac{b_j}{c_j} \al(c_j)\right)\left( 1 + O(e^{-\mathfrak{c}n})\right),
\end{equation}
 as $n \to \infty$, where 	\begin{equation}\label{al al}
 \al(z):= \exp \left[ \frac{1}{2 \pi i } \int_{\T} \frac{\ln(\phi(\tau))}{\tau-z}d\tau \right],
 \end{equation} $G[\phi]$ and $E[\phi]$ are given by \eqref{G and E} and  $\mathfrak{c}$ is some positive constant.
\end{corollary}
\begin{proof}
	Note that \eqref{228} immediately follows from \eqref{psi=1/z-c}, \eqref{border,phi,const} and \eqref{linear-combination}; and then we get \eqref{229} as a direct consequence of \eqref{Szego Theorem}, \eqref{G and E}, \eqref{X in terms of R exact} and \eqref{R asymp}.
\end{proof}


\subsection{Bordered Toeplitz determinants of the type $D^B_{N}[\phi;q\phi]$}

Now we turn our attention to the bordered Toeplitz determinants where the border symbol is given by $q \phi$, $q$ being a rational function with simple poles. Let us start with proving a fundamental identity relating one such bordered Toeplitz determinant to the pure Toeplitz Riemann-Hilbert data.

\begin{lemma}\label{lem2.2} Let $\psi_0:=q_0\phi$, where $q_0$ is defined in \eqref{q1 and q2}, with $c \neq 0$. Then the bordered Toeplitz determinant $D^{B}_n[\phi;\psi_0]$ can be written in terms of the following data from the solution of the X-RHP:
\begin{equation}\label{bdToep-RHP-q0}
    D^{B}_{n+1}[\phi;\psi_0] = - \frac{1}{c}D_{n+1}[\phi]+ \frac{1}{c} D_{n}[\phi] X_{12}(c,n),
\end{equation}
where $D_{n}[\phi]$ is given by \eqref{ToeplitzDet} and $X_{12}$ is the $12$ entry of the solution to  \textbf{RH-X1} through  \textbf{RH-X3}.
\end{lemma}

\begin{proof}
Note that \begin{equation*}
    \psi_{0,j} = \int_{\T} z^{-j} \psi_0(z) \frac{\dd z}{2 \pi \ic z} = \int_{\T} z^{-j}  \frac{1}{z(z-c)} \phi(z) \frac{\dd z}{2 \pi \ic}, 
\end{equation*}thus
\begin{equation}
  \psi_{0,j}  = -\frac{1}{c}  \int_{\T} z^{-j}  \phi(z) \frac{\dd z}{2 \pi \ic z} + \frac{1}{c}  \int_{\T}   \frac{z^{-j}\phi(z)}{(z-c)}  \frac{\dd z}{2 \pi \ic}  = -\frac{1}{c}\phi_j +
  \frac{1}{c}  \int_{\T}   \frac{z^{-j}\phi(z)}{(z-c)}  \frac{\dd z}{2 \pi \ic}.
\end{equation}
Now, observe that
\begin{equation*}
    D^{B}_{n+1}[\phi,\psi_0] = \det \begin{pmatrix}
\phi_0 & \phi_{-1} & \cdots & \phi_{-n} \\
\phi_1 & \phi_{0} & \cdots & \phi_{-n+1} \\
\vdots & \vdots & \cdots & \vdots \\
\phi_{n-1} & \phi_{n-2} & \cdots & \phi_{-1} \\
\psi_{0,n} & \psi_{0,n-1} & \cdots & \psi_{0,0}
\end{pmatrix} \end{equation*} \begin{equation*}
 =  \frac{1}{c}  \det \begin{pmatrix}
\phi_0 & \phi_{-1} & \cdots & \phi_{-n} \\
\phi_1 & \phi_{0} & \cdots & \phi_{-n+1} \\
\vdots & \vdots & \cdots & \vdots \\
\phi_{n-1} & \phi_{n-2} & \cdots & \phi_{-1} \\
-\phi_n + \int_{\T}   \frac{z^{-n}\phi(z)}{(z-c)}  \frac{\dd z}{2 \pi \ic} & -\phi_{n-1} + \int_{\T}   \frac{z^{-n+1}\phi(z)}{(z-c)}  \frac{\dd z}{2 \pi \ic} & \cdots & -\phi_0 + \int_{\T}   \frac{\phi(z)}{(z-c)}  \frac{\dd z}{2 \pi \ic}
\end{pmatrix} =
\end{equation*}
\begin{equation}\label{AA}
    -\frac{1}{c} \det \begin{pmatrix}
\phi_0 & \phi_{-1} & \cdots & \phi_{-n} \\
\phi_1 & \phi_{0} & \cdots & \phi_{-n+1} \\
\vdots & \vdots & \cdots & \vdots \\
\phi_{n-1} & \phi_{n-2} & \cdots & \phi_{-1} \\
\phi_n  & \phi_{n-1} & \cdots & \phi_0
\end{pmatrix} + \frac{1}{c}\det \begin{pmatrix}
\phi_0 & \phi_{-1} & \cdots & \phi_{-n} \\
\phi_1 & \phi_{0} & \cdots & \phi_{-n+1} \\
\vdots & \vdots & \cdots & \vdots \\
\phi_{n-1} & \phi_{n-2} & \cdots & \phi_{-1} \\
 \int_{\T}   \frac{z^{-n}\phi(z)}{(z-c)}  \frac{\dd z}{2 \pi \ic} &  \int_{\T}   \frac{z^{-n+1}\phi(z)}{(z-c)}  \frac{\dd z}{2 \pi \ic} & \cdots &  \int_{\T}   \frac{\phi(z)}{(z-c)}  \frac{\dd z}{2 \pi \ic}
\end{pmatrix}.
\end{equation}
Now note that, using \eqref{Toeplitz OP 1} and \eqref{ka_n} we have

\begin{equation}
    \ka_n^{-1} \ze^{-n}Q_n(\ze) = \frac{1}{D_n[\phi]} \det \begin{pmatrix}
\phi_0 & \phi_{-1} & \cdots & \phi_{-n} \\
\phi_1 & \phi_{0} & \cdots & \phi_{-n+1} \\
\vdots & \vdots & \cdots & \vdots \\
\phi_{n-1} & \phi_{n-2} & \cdots & \phi_{-1} \\
\ze^{-n} & \ze^{-n+1} & \cdots & 1
\end{pmatrix}.
\end{equation}
Combining this equation with \eqref{Toeplitz-OP-solution} yields
\begin{equation}\label{X12 bordered}
X_{12}(z;n) = \frac{1}{D_n[\phi]} \det \begin{pmatrix}
\phi_0 & \phi_{-1} & \cdots & \phi_{-n} \\
\phi_1 & \phi_{0} & \cdots & \phi_{-n+1} \\
\vdots & \vdots & \cdots & \vdots \\
\phi_{n-1} & \phi_{n-2} & \cdots & \phi_{-1} \\
\int_{\T}  \frac{\ze^{-n} \phi(\ze)}{\ze-z} \frac{\dd \ze}{2\pi \ic} & \int_{\T}  \frac{\ze^{-n+1} \phi(\ze)}{\ze-z} \frac{\dd \ze}{2\pi \ic} & \cdots & \int_{\T}  \frac{\phi(\ze)}{\ze-z} \frac{\dd \ze}{2\pi \ic}
\end{pmatrix}.
\end{equation}
Thus, using \eqref{AA} and \eqref{X12 bordered} we arrive at \eqref{bdToep-RHP-q0}. 
\end{proof}

\begin{corollary}\label{cor 2}
	We have
	\begin{equation}\label{235}
	D^B_{n+1}\left[\phi;\di\left( a+\sum^{m}_{j=1}\frac{b_j z}{z-c_j}\right)\phi \right] =aD_{n+1}[\phi]+ D_n[\phi] \sum^{m}_{j=1}b_j X_{12}(c_j;n),
	\end{equation}
	and for a Szeg{\H o} type $\phi$
	\begin{equation}\label{236}
	D^B_{n+1}\left[\phi;\di\left( a+\sum^{m}_{j=1}\frac{b_j z}{z-c_j}\right)\phi \right] = G[\phi]^{n+1} E[\phi] \left(a+\frac{1}{G[\phi]}\sum^{m}_{j=1  \atop |c_j|<1 }b_j \al(c_j)\right)\left( 1 + O(e^{-\mathfrak{c}n})\right),
	\end{equation}
	as $n \to \infty$, where $\al$ is defined in \eqref{al al},  $G[\phi]$ and $E[\phi]$ are given by  \eqref{G and E}, and  $\mathfrak{c}$ is some positive constant.
\end{corollary}
\begin{proof}
	\eqref{235} directly follows from \eqref{bdToep-RHP-q0}, \eqref{border,phi,phi}, and \eqref{linear-combination}. For the asymptotic statement, notice that from  \eqref{R_k's are small},  \eqref{X in terms of R exact} and \eqref{R asymp} we have \begin{equation}\label{237}
	X_{12}(c;n) = \begin{cases}
	\al(c) (1+O(e^{-\mathfrak{c}n})), & |c|<1, \\
	R_{1,12}(c;n)\al^{-1}(c)c^{-n}(1+O(e^{-\mathfrak{c}n})), & |c|>1,
	\end{cases}
	\end{equation}
	 as $n \to \infty$,	where $R_{1,12}$ is given by  \eqref{R1}. Now \eqref{236} follows from \eqref{Szego Theorem}, \eqref{G and E}, \eqref{235}, \eqref{237}, and \eqref{R_k's are small}.
\end{proof}
Now, we prove the analogue of Lemma \ref{lem2.2} for $c=0$.

\begin{lemma}\label{analogue of lemma 2.2 for c=0}
We have the following identity\footnote{Throughout the paper we occasionally use $\tilde{f}(z)$, to denote $\di f(z^{-1})$.} \begin{equation}
    D^B_{n}[\phi;\frac{1}{z}\phi] = - D^B_{n+1}[\Tilde{\phi};z],
\end{equation} 
and hence for a Szeg{\H o} type $\phi$
\begin{equation}\label{phi,z^-1phi asymp}
     D^B_{n}[\phi;\frac{1}{z}\phi] = \frac{1}{2\pi \ic} G[\phi]^{n} E[\phi]\left( \int_{\T} \ln(\Tilde{\phi}(\tau))  \dd \tau + O(e^{-\mathfrak{c}n}) \right),
\end{equation}
as $n \to \infty$, where $G[\phi]$ and $E[\phi]$ are given by  \eqref{G and E}, and  $\mathfrak{c}$ is some positive constant.
\end{lemma}

\begin{proof}
Note that
\begin{align*}
	D^B_{n+1}[\Tilde{\phi};z] 
	& = \det \begin{pmatrix}
	\Tilde{\phi}_0 & \Tilde{\phi}_{-1} & \cdots  & \Tilde{\phi}_{-n+2}  & \Tilde{\phi}_{-n+1} &   \Tilde{\phi}_{-n} \\
	\Tilde{\phi}_1 & \Tilde{\phi}_{0} & \cdots  & \Tilde{\phi}_{-n+3} & \Tilde{\phi}_{-n+2} & \Tilde{\phi}_{-n+1} \\
	\vdots & \vdots & \cdots  & \vdots & \vdots & \vdots \\
	\Tilde{\phi}_{n-1} & \Tilde{\phi}_{n-2} & \cdots & \Tilde{\phi}_{1} & \Tilde{\phi}_{0} &  \Tilde{\phi}_{-1} \\
	0 & 0 & \cdots & 0 & 1 & 0 
	\end{pmatrix} \\ 
	& =  \det \begin{pmatrix}
	\phi_0 & \phi_{1} & \cdots  & \phi_{n-2}  & \phi_{n-1} &   \phi_{n} \\
	\phi_{-1} & \phi_{0} & \cdots  & \phi_{n-3} & \phi_{n-2} & \phi_{n-1} \\
	\vdots & \vdots & \cdots  & \vdots & \vdots & \vdots \\
	\phi_{-n+1} & \phi_{-n+2} & \cdots & \phi_{-1} & \phi_{0} &  \phi_{1} \\
	0 & 0 & \cdots & 0 & 1 & 0 
	\end{pmatrix} \\ 
	& =  - \det \begin{pmatrix}
	\phi_0 & \phi_{1} & \cdots  & \phi_{n-2}  &    \phi_{n} \\
	\phi_{-1} & \phi_{0} & \cdots  & \phi_{n-3} &  \phi_{n-1} \\
	\vdots & \vdots & \cdots  & \vdots &  \vdots \\
	\phi_{-n+1} & \phi_{-n+2} & \cdots & \phi_{-1} &   \phi_{1}  
	\end{pmatrix} \\
	& =
	- \det \begin{pmatrix}
	\phi_0 & \phi_{-1} & \cdots  & \phi_{-n+1}  \\ \phi_1 & \phi_{0} & \cdots  & \phi_{-n+2} \\
	\vdots & \vdots & \cdots  & \vdots  \\
	\phi_{n-2} & \phi_{n-3} & \cdots & \phi_{-1} \\
	\phi_{n} & \phi_{n-1} & \cdots & \phi_{1}  
	\end{pmatrix} 
	= -D^B_{n}[\phi;\frac{1}{z}\phi], 
\end{align*}

because the $j$-th Fourier coefficient of $z^{-1}\phi(z)$ is $\phi_{j+1}$. Now \eqref{phi,z^-1phi asymp} immediately follows from \eqref{phi,z asymp1} and the fact that $D_n[\phi]=D_n[\Tilde{\phi}]$.
\end{proof}

\begin{lemma}\label{psi=z^k phii}
	For $\psi = z^k \phi$, $k=0,1,\cdots, n$, we have
	\begin{equation}\label{psi=z^kphi}
	D^B_n[\phi,z^k\phi] = \begin{cases} 
	D_n[\phi], & k=0, \\
	0, & k=1,\cdots,n-1,\\
	(-1)^{n-1}D_n[z\phi], & k=n.
	\end{cases}
	\end{equation}
\end{lemma}
\begin{proof}
	Note that $[z^k\phi]_j = \phi_{j-k}$ and thus, $D^B_n[\phi,z^k\phi]$ has two identical columns for $k=1,\cdots,n-1$, and \eqref{psi=z^kphi} is obvious for  $k=0$. For $k=n$, \eqref{psi=z^kphi} follows immediately if one moves the border column to the first column by making $n-1$ swaps of adjacent columns.
\end{proof}

Theorem \ref{main thm} is now proven by combining lemmas \ref{psi=z},  \ref{analogue of lemma 2.2 for c=0}, \ref{psi=z^k phii} and the corollaries  \ref{cor 1} and \ref{cor 2} via \eqref{linear-combination}.


\subsection{Ising model next-to-diagonal correlations.} In this section we focus on the specific symbols $\widehat{\phi}$ and $\widehat{\psi}$, respectively given by \eqref{Isingphi} and \eqref{hat psi}, corresponding to the next-to-diagonal correlations in the Ising model via \eqref{BT&NTD}. For a derivation of \eqref{hat psi} from the formulae in \cite{YP} see Section \ref{Simplifying AuYangPerk}. This is clear that in the low temperature regime ($k>1$) the symbols $\widehat{\phi}$ and $\widehat{\psi}$ fit the class of symbols considered in Theorem \ref{main thm}. Indeed, comparing \eqref{hat psi} with \eqref{general psi} we can find the corresponding parameters:
\begin{equation}
\begin{split}
\begin{aligned}
	& m =1,              & a&=b_0=\hat{b}_0=0,        \qquad & b_1 &= \frac{C_v}{S_v}, \\
	&  c_1=c_*\equiv -\frac{S_h}{S_v},         \qquad  & \hat{a}&=\frac{C_h}{S_h} ,          & \hat{b}_1&=-\frac{C_h}{S_h}.
\end{aligned}
\end{split}
\end{equation}
Therefore the constant $F[\widehat{\phi};\widehat{\psi}]$ given by \eqref{ConstantF2} simplifies to 
\begin{equation}\label{FISING}
	F[\widehat{\phi};\widehat{\psi}] = \begin{cases}
	\di \frac{C_v}{S_v} \al(c_*), &  J_v>J_h, \\[11pt]
	\di \frac{C_h}{S_h} \al(c_*), & J_v<J_h.
	\end{cases}
\end{equation}
where we have used \[ |c_*|\lessgtr1 \iff J_h \lessgtr J_v. \]
Now let us compute $\al(c_*)$. We observe that \begin{equation}
\widehat{\phi}(z)=e^{i\pi/2} k^{1/2}(z-k)^{-1/2}(z-k^{-1})^{1/2}z^{-1/2},  
\end{equation}where the branches of the roots all have arguments from $0$ to $2\pi$.  Recalling the expression \eqref{al al} for $\al$, we can compute $\al(c_*)$ by a simple contour integration (deform the integral on $\T$ to the interval $[0,k^{-1}]$, and note that we get a residue term when $-1<c_*<0$ ). We eventually arrive at
\begin{equation}\label{AL(c_*)}
\al(c_*) = \begin{cases} \di \frac{S_v}{C_v}, & J_v>J_h, \\[11pt]
\di \frac{S_h}{C_h}, & J_v<J_h.
\end{cases}
\end{equation}
This can also be seen in a more straightforward way by recalling \eqref{alpha.phi.pm}, which in the Ising case amounts to: 
\begin{equation}
	\al(z) =  \begin{cases}
		\di \frac{1}{\sqrt{1-k^{-1}z}}, & |z|<1, \\[11pt]
		\di \frac{1}{\sqrt{1-k^{-1}z^{-1}}}, & |z|>1.
	\end{cases}
\end{equation}
Combining \eqref{FISING} and \eqref{AL(c_*)} yields
\begin{equation}
	F[\widehat{\phi};\widehat{\psi}] = 1, \qquad J_h \neq J_v.
\end{equation}
This concludes the proof of Theorem \ref{th2} in the anisotropic case by recalling that $G[ \widehat{\phi} \ ]=1$ and $E[ \widehat{\phi} \ ]=(1-k^{-2})^{1/4}$.

\begin{remark} \normalfont
Notice that the next-to-diagonal long range order in the isotropic case $J_h=J_v$ deserves special attention as $|c_*|=1$. It is important to notice that while it seems that the function $\widehat{\psi}(z)$ has a pole at  $z=c_*$ it actually has a removable singularity there. In other words, $\widehat{\psi}(z)$ (as well as $\widehat{\phi}(z)$)
are analytic in a neighborhood of the unit circle $\T$ irrespective of the value of $c_*$. To be precise both functions are analytic in $z$ for $k^{-1}<|z|<k$. However, the splitting of $\widehat{\psi}(z)$ done in the proof of Theorem  \ref{main thm} would re-introduce this pole at $z=c_*$ in both terms and render the proof invalid when $|c_*|=1$. To circumvent this problem we resort to a ``deformation trick''.
For a fixed $0<\ep<\di \frac{k-k^{-1}}{k+k^{-1}}$, let $ \di \frac{k^{-1}}{1-\ep} < \rho < \di \frac{k}{1+\ep}$ and introduce the functions
$$
\widehat{\phi}_{\rho}(z):=\widehat{\phi}(\rho z),\qquad \widehat{\psi}_{\rho}(z):=\widehat{\psi}(\rho z),\qquad 
$$ where the condition on $\rho$ ensures that both functions are analytic in the $\ep$-neighborhood of $\T$: $\{z: 1-\ep<|z|<1+\ep\}$. The $n$-th Fourier coefficient of these new functions differs from $n$-the Fourier coefficient of the original functions by a factor 
$\rho^n$. For this reason, we have that 
$$
D_N[\widehat{\phi}]=D_N[\widehat{\phi}_\rho],\qquad D_N^B[\widehat{\phi};\widehat{\psi}]=D_N^B[\widehat{\phi}_\rho;\widehat{\psi}_\rho].
$$
Indeed, the underlying matrices are related to each other by appropriate multiplication of diagonal matrices.
Thus the derivation of the determinant asymptotics can be based on the pair $(\widehat{\phi}_\rho;\widehat{\psi}_\rho)$ rather than the pair $(\widehat{\phi};\widehat{\psi})$.
These functions are clearly of a similar form as the original ones. The crucial point however is that 
when we consider $\widehat{\psi}_\rho$ and split it into two parts, each part  has a pole at $z=c_*/\rho$ rather than at $z=c_*$. 
Thus Theorem  \ref{main thm}  is applicable to $(\widehat{\phi}_\rho;\widehat{\psi}_\rho)$ whenever $|c_*|\neq \rho$.
As we can choose this $\rho$ within at least a small range, $k^{-1}<\rho<k$, the asymptotic results concerning these functions remains true
also when $|c_*|=1$. Notice that since the determinant $D_N^B[\widehat{\phi}_\rho;\widehat{\psi}_\rho]$ is $\rho$-independent, in particular, its leading order asymptotics given by Theorem \ref{main thm} is also $\rho$-independent. This can also be checked directly by looking at the terms on the right hand side of \eqref{main formula}. To this end, we obviously have from \eqref{G and E} that $G[\widehat{\phi}_\rho] = G[\widehat{\phi}]$ and $E[\widehat{\phi}_\rho] = E[\widehat{\phi}]$. Also the role of $\al$ on the right hand side of \eqref{ConstantF2} is now played by $\al_{\rho}$ which satisfies $\al_{\rho,+}(z)=\al_{\rho,-}(z)\widehat{\phi}_\rho(z) $ and is explicitly given by  $$
\al_{\rho}(z) =  \begin{cases}
\di \frac{1}{\sqrt{1-k^{-1}\rho z}}, & |z|<1, \\[11pt]
\di \frac{1}{\sqrt{1-k^{-1} \rho^{-1} z^{-1}}}, & |z|>1.
\end{cases}
$$  Therefore from \eqref{ConstantF2} it can be directly checked that  $F[\widehat{\phi}_\rho;\widehat{\psi}_\rho]=1$.
\end{remark}
\subsection {Proof of Theorem \ref{th22}}
Based on \eqref{linear-combination}, \eqref{hat psi}, \eqref{228}, and \eqref{235}, \color{black} the bordered Toeplitz determinant representing the Ising correlation function 
$\langle \sigma_{0,0}\sigma_{N-1,N} \rangle$  satisfies the following
relation
\begin{equation}\label{start1}
 \frac{D^{B}_{N}[\widehat{\phi}; \widehat{\psi}]}{D_{N-1}[\widehat{\phi} \ ]} = \di \frac{C_v}{S_v} X_{12}(c_*;N-1)
+  
\begin{cases}
c^{-N+1}_* \di \frac{C_h}{S_h} X_{11}(c_*;N-1), &  |c_*| > 1, \\[11pt]
0, & |c_*| <1.
\end{cases}
\end{equation}
Also, for the Toeplitz determinant $D_{N-1}[\widehat{\phi} \ ]$ we have that
\begin{equation}\label{start2}
\ln D_{N}[\widehat{\phi} \ ] = \ln E[\widehat{\phi} \ ]  -\sum _{n=N}^{\infty}\ln \kappa^{-2}_{n}
=\frac{1}{4}\ln\Bigl(1 - k^{-2}\Bigr) -\sum _{n=N}^{\infty}\ln \kappa^{-2}_{n}  ,
\end{equation}
where 
\begin{equation}\label{start3}
\kappa^{-2}_n = X_{12}(0),
\end{equation}
and $X(z) \equiv X(z; N-1)$ is the solution of the $X$ - Riemann-Hilbert problem, {\bf{RH-X1}} - {\bf{RH-X3}}
generated by the weight $\widehat{\phi}(z)$ and corresponding to $n = N-1$. Equations (\ref{start1}) - (\ref{start3}) show us that in order to calculate the correction to the 
leading term, i.e., $(1-k^{-2})^\frac{1}{4}$, to the  determinant $D^{B}_{N}[\widehat{\phi}; \widehat{\psi}]$
we need the high terms in the  estimation of  the solution $X(z)$ of the Riemann-Hilbert problem. 

The asymptotic analysis of the Riemann-Hilbert problem {\bf{RH-X1}} - {\bf{RH-X3}} is presented
in detail in Section \ref{Appendix Y-RHP} and for its solution $X(z;n)$ we have formula 
(\ref{X in terms of R exact}) where $R(z;n)$ is the solution of the small norm Riemann-Hilbert
problem {\bf RH-R1} - {\bf RH-R3}. We have already used this formula and the estimate 
(\ref{R asymp}) in the proof of Theorem \ref{main thm}. Now, we need more terms in (\ref{R asymp}).
These are given by the second term, $R_2(z;n)$ in the iterative series (\ref{Appendix pure toep R series}).
From (\ref{Recursive R_k Def}), (\ref{JR-I}), and (\ref{R1}) that 

\begin{equation}\label{R2}
R_2(z;n)=  \begin{pmatrix}
-\di \frac{1}{4\pi^2}\int_{\Ga_1}b(\mu)\left[\int_{\Ga_0}a(\tau)\frac{d\tau}{\tau-\mu}\right]\frac{d\mu}{\mu-z}
&0\\ \\
0&-\di \frac{1}{4\pi^2}\int_{\Ga_0}a(\mu)\left[\int_{\Ga_1}b(\tau)\frac{d\tau}{\tau-\mu}\right]\frac{d\mu}{\mu-z}
 \end{pmatrix}, \quad
z\in \C \setminus \Sigma_R,
\end{equation}
where we have introduced the notations,
$$
a(z;n) = -z^n \widehat{\phi}^{-1}(z) \alpha^2(z), \quad \mbox{and}\quad
b(z;n) = z^{-n} \widehat{\phi}^{-1}(z) \alpha^{-2}(z),
$$ 
and we also remind that
$$
\al(z) =\exp \left[ \frac{1}{2 \pi i } \int_{\T} \frac{\ln(\widehat{\phi}(\tau))}{\tau-z}d\tau \right]=  \begin{cases}
\di \frac{1}{\sqrt{1-k^{-1}z}}, & |z|<1, \\[11pt]
\di \frac{1}{\sqrt{1-k^{-1}z^{-1}}}, & |z|>1.
\end{cases}
$$
Formula (\ref{R2}) in conjunction with the equations  \eqref{R1} and (\ref{X in terms of R exact}) yield the following
estimations for the relevant matrix entries of $X(z)$ in the regions $\Om_0$, $\Om_1$, $\Om_2$, and $\Om_{\infty}$, (see Figure \ref{S_contour}):
\begin{equation}\label{X1201}
X_{12}(z;n) = \alpha(z)\left( 1 -\di \frac{1}{4\pi^2}\int_{\Ga_1}b(\mu;n)\left[\int_{\Ga_0}a(\tau;n)
\frac{d\tau}{\tau-\mu}\right]\frac{d\mu}{\mu-z}+ O(\rho^{-4n})\right), \quad z \in \Omega_0\cup\Omega_1,
\end{equation}

\begin{equation}\label{X122infty}
X_{12}(z;n) = z^{-n}\alpha^{-1}(z)\left(\frac{1}{2\pi i} \int_{\Ga_0}a(\tau;n) 
\frac{d\tau}{\tau-z}+  O\left(\frac{\rho^{-3n}}{1+|z|}\right)\right), \quad z \in \Omega_2\cup\Omega_{\infty},
\end{equation}

\begin{equation}\label{X110}
X_{11}(z;n) = -\alpha^{-1}(z)\left(\frac{1}{2\pi i} \int_{\Ga_0}a(\tau;n) 
\frac{d\tau}{\tau-z}+  O(\rho^{-3n})\right), \quad z \in \Omega_0,
\end{equation}

$$
X_{11}(z;n) = z^n\alpha(z)\widehat{\phi}^{-1}(z)\left( 1 -\di \frac{1}{4\pi^2}\int_{\Ga_1}b(\mu;n)\left[\int_{\Ga_0}a(\tau;n)
\frac{d\tau}{\tau-\mu}\right]\frac{d\mu}{\mu-z}+ O(\rho^{-4n})\right)
$$
\begin{equation}\label{X111}
-\alpha^{-1}(z)\left(\frac{1}{2\pi i} \int_{\Ga_0}a(\tau;n) 
\frac{d\tau}{\tau-z}+  O(\rho^{-3n})\right),
 \quad z \in \Omega_1,
\end{equation}

$$
X_{11}(z;n) = z^n\alpha(z)\left( 1 -\di \frac{1}{4\pi^2}\int_{\Ga_1}b(\mu;n)\left[\int_{\Ga_0}a(\tau;n)
\frac{d\tau}{\tau-\mu}\right]\frac{d\mu}{\mu-z}+  O(\rho^{-4n})\right)
$$
\begin{equation}\label{X112}
-\alpha^{-1}(z)\widehat{\phi}^{-1}(z)\left(\frac{1}{2\pi i} \int_{\Ga_0}a(\tau;n) 
\frac{d\tau}{\tau-z}+  O(\rho^{-3n})\right),
 \quad z \in \Omega_2,
\end{equation}

\begin{equation}\label{X11infty}
X_{11}(z;n) = z^n\alpha(z)\left( 1 -\di \frac{1}{4\pi^2}\int_{\Ga_1}b(\mu;n)\left[\int_{\Ga_0}a(\tau;n)
\frac{d\tau}{\tau-\mu}\right]\frac{d\mu}{\mu-z}+  O\left(\frac{\rho^{-4n}}{1+|z|}\right)\right), \quad z \in \Omega_\infty.
\end{equation}

In the above equations, it is assumed that the circle  $\Gamma_1$ is centered at $z =0$ and has radius $\rho$,
the circle $\Gamma_0$ is centered at $z =0$ and  has  radius $\rho^{-1}$, and the inequality 
\begin{equation}\label{rhok}
k^{\frac{2}{3}} < \rho <  k
\end{equation}
holds. Note that the last equation implies that
\begin{equation}\label{rhok2}
\rho^{-3n} = k^{-(2 +\delta)n},
\end{equation}
where
$$
\delta = \frac{3\ln\rho- 2\ln k}{\ln k} > 0.
$$

Let us first  apply these formulae to the evolution of the higher term in the asymptotics
of the pure Toeplitz determinant $D_{N}[\widehat{\phi} \ ]$. To this end, taking into account (\ref{start2}),
(\ref{start3})  we need first to estimate the value $X_{12}(0)$. According to (\ref{X1201}), we have that
\begin{equation}\label{X1200}
X_{12}(0;n) = \alpha(0)\left( 1 -\di \frac{1}{4\pi^2}\int_{\Ga_1}b(\mu;n)\left[\int_{\Ga_0}a(\tau;n)
\frac{d\tau}{\tau-\mu}\right]\frac{d\mu}{\mu}+ O(\rho^{-4n})\right).
\end{equation}
\begin{proposition}\label{intas} The following estimates take place,
\begin{equation}\label{intas1}
\begin{split}
\frac{1}{2\pi i} \int_{\Ga_0}a(\tau;n) 
\frac{d\tau}{\tau-z} =
-\frac{1}{\sqrt{\pi}}\alpha^{2}(k^{-1})\frac{\sqrt{k-k^{-1}}}{k}
\frac{1}{k^{-1} -z}k^{-n -1/2}n^{-1/2} \left(1+O\left(\frac{1}{n}\right)\right), \quad n \to \infty, 
\end{split}
\end{equation}
for all $|z| > \rho^{-1}$, and
\begin{equation}\label{intas2}
\di \frac{1}{4\pi^2}\int_{\Ga_1}b(\mu;n)\left[\int_{\Ga_0}a(\tau;n)
\frac{d\tau}{\tau-\mu}\right]\frac{d\mu}{\mu-z} =
-\frac{1}{2\pi}
\frac{1}{k^{-1} - k}\frac{1}{k-z}
k^{-2n}n^{-2}\left( 1+ O\left(\frac{1}{n}\right)\right), \quad n \to \infty,
\end{equation}
for all $|z| < \rho$.
\end{proposition}

\begin{proof}
Consider first the single, $a$ - integral. It can be deformed to the integral over the 
segment $[0, k^{-1}]$ so that we would have,
\begin{equation}\label{131313}
	\frac{1}{2\pi i} \int_{\Ga_0}a(\tau;n) 
	\frac{d\tau}{\tau-z} = -\frac{1}{\pi } \int_{0}^{k^{-1}}\tau^n(k^{-1} - \tau)^{-1/2}\Phi(\tau)d\tau,
\end{equation}
where
$$
\Phi(\tau) = \sqrt{\tau k^{-1}}\sqrt{k-\tau}\frac{\alpha^{2}(\tau)}{\tau - z}
$$
is holomorphic at  $\tau = k^{-1}$. Let
\begin{equation}\label{141414}
	\Phi(\tau) = \sum_{l=0}^{\infty}d_l(\tau-k^{-1})^l, \quad d_0 = \Phi(k^{-1}) = 
	\frac{1}{k}\sqrt{k-k^{-1}}\frac{\alpha^{2}(k^{-1})}{k^{-1} - z},
\end{equation}
be the Taylor series of $\Phi(\tau)$ at $\tau = k^{-1}$. Then, according to the standard Watson Lemma type arguments, we arrive at the asymptotic  formula,
\begin{equation}\label{hN1}
\frac{1}{2\pi i} \int_{\Ga_0}a(\tau;n) 
\frac{d\tau}{\tau-z} \sim -\frac{1}{\pi} \sum_{l=0}^{\infty}d_l\int_{0}^{k^{-1}}\tau^n(k^{-1}-\tau)^{l-1/2}d\tau,
\end{equation}
and this asymptotic is uniform in any compact subset of $\{z: |z| > \rho^{-1}\}$ and, in particular, for $z \in \Gamma_1$.
For the integrals in the right hand side of (\ref{hN1}), we have,
\begin{equation}\label{151515}
	\begin{split}
	\int_{0}^{k^{-1}}\tau^n(k^{-1}-\tau)^{l-1/2}d\tau  & = k^{-n-l-1/2}\int_{0}^{1}t^n(1-t)^{l-1/2}dt = k^{-n-l-1/2}B(n+1, l+1/2) \\ & = k^{-n-l-1/2}\frac{\Gamma(n+1)\Gamma(l+1/2)}{\Gamma(n+l+3/2)} = k^{-n-l-1/2}\Gamma(l+1/2)n^{-l-1/2}\left(1+O\left(\frac{1}{n}\right)\right),
	\end{split}
\end{equation}
as $n \to \infty$, and hence
\begin{equation}\label{hN2}
\frac{1}{2\pi i} \int_{\Ga_0}a(\tau;n) 
\frac{d\tau}{\tau-z} = 
-\frac{1}{\sqrt{\pi}}\alpha^{2}(k^{-1})\frac{\sqrt{k-k^{-1}}}{k}
\frac{1}{k^{-1} -z}k^{-n -1/2}n^{-1/2}\left(1+O\left(\frac{1}{n}\right)\right),
\end{equation}
as $n \to \infty$, uniformly in any compact subset of $\{z: |z| > \rho^{-1}\}$. This is the estimate (\ref{intas1}).

Consider now the double integral (\ref{intas2}). In view of (\ref{hN2}) we have at once that{\footnote{ We are also
taking into account that $\int_{\Gamma_1}b(\mu)f(\mu) d\mu = O(k^{-n})$ for any bounded $f(\mu)$ analytic in the annulus, $\rho -\epsilon
<|\mu| < k+\epsilon$.}} 
\begin{equation}\label{121212}\begin{split}
	\di \frac{1}{4\pi^2}\int_{\Ga_1}b(\mu;n)\left[\int_{\Ga_0}a(\tau;n)
\frac{d\tau}{\tau-\mu}\right]\frac{d\mu}{\mu-z} & = -\frac{i}{2\pi^{3/2}}\alpha^{2}(k^{-1})\frac{\sqrt{k-k^{-1}}}{k}k^{-n -1/2}n^{-1/2}
 \\ & \times \left[\int_{\Gamma_1}b(\mu;n)\frac{d\mu}{(k^{-1} - \mu)(\mu - z)} + O\left(\frac{k^{-n}}{n}\right)\right], \quad n \to \infty.
\end{split}
\end{equation}
Applying to the $b$ -integral in the right hand side of the last equation the same arguments as we 
used for derivation of the estimate (\ref{hN2}), we  obtain that
 \begin{equation}\label{hN3}
 \int_{\Gamma_1}b(\mu;n)\frac{d\mu}{(k^{-1} - \mu)(\mu - z)} = -i\sqrt{\pi}\alpha^{-2}(k)\frac{1}{(k-z)\sqrt{k-k^{-1}}}
\frac{1}{k^{-1} -k}k^{-n +3/2}n^{-3/2}\left(1+O\left(\frac{1}{n}\right)\right),
\end{equation}
as $
n \to \infty
$, uniformly in any compact subset of $\{z: |z| <  \rho\}$. This in turns yields the estimate, 
$$
\di \frac{1}{4\pi^2}\int_{\Ga_1}b(\mu;n)\left[\int_{\Ga_0}a(\tau;n)
\frac{d\tau}{\tau-\mu}\right]\frac{d\mu}{\mu-z} = -\frac{1}{2\pi}
\frac{\alpha^2(k^{-1})\alpha^{-2}(k)}{k^{-1} - k}\frac{1}{k-z}
k^{-2n}n^{-2}\left( 1+ O\left(\frac{1}{n}\right)\right), \quad n \to \infty,
$$
which, taking into account that 
$$
\alpha^{-1} (z)\alpha(z^{-1}) =1,
$$
implies (\ref{intas2})
\end{proof} 

Using (\ref{X1200}), (\ref{intas2}) and taking into account that $\alpha(0) =1$ we conclude  that
\begin{equation}\label{kappaas1}
\kappa^{-2} = X_{12}(0;n)
= 1 +\frac{1}{2\pi}
\frac{1}{1 - k^2}
k^{-2n}n^{-2}\left( 1+ O\left(\frac{1}{n}\right)\right), \quad n \to \infty,
\end{equation}
and, also,
\begin{equation}\label{kappaas2}
\ln\kappa^{-2} _n = \frac{1}{2\pi}
\frac{1}{1 - k^2}
k^{-2n}n^{-2}\left( 1+ O\left(\frac{1}{n}\right)\right), \quad n \to \infty.
\end{equation}
In view of  (\ref{start2}), we need now the asymptotics of the sum 
$$
\sum_{n=N}^{\infty}\frac{k^{-2n}}{n^p}, \quad p = 2, 3.
$$
This can be easily done by the summation by parts.  Indeed, put
$$
S_n = \sum_{l=n}^{\infty}k^{-2l} \equiv k^{-2n}\frac{1}{1-k^{-2}}.
$$
Then we have,
$$
\sum_{n=N}^{\infty}\frac{k^{-2n}}{n^p}= \sum_{n=N}^{\infty}(S_{n} -S_{n+1})\frac{1}{n^p}
= \sum_{n=N}^{\infty}S_{n}\frac{1}{n^p} - \sum_{n=N}^{\infty}S_{n+1}\frac{1}{n^p}
$$
$$
= S_{N}\frac{1}{N^p} + \sum_{n=N+1}^{\infty}S_{n}\frac{1}{n^p} - \sum_{n=N}^{\infty}S_{n+1}\frac{1}{n^p}
= S_N\frac{1}{N^p} + \sum_{n=N}^{\infty}S_{n+1}\left(\frac{1}{(n+1)^p} - \frac{1}{n^p}\right)
$$
\begin{equation}\label{sumas1}
= \frac{1}{1-k^{-2}}N^{-p}k^{-2N} + O\Bigl(N^{-p-1} k^{-2N}\Bigr).
\end{equation}
Combining (\ref{sumas1}) with  (\ref{kappaas2}) and (\ref{start2}) we arrive at the final formula
for the asymptotics of the Toeplitz determinant $D_N[\widehat{\phi} \ ]$ with the explicit second term,
\begin{equation}\label{DNasRH}
D_N[\widehat{\phi} \ ] = (1-k^{-2})^{1/4}\left(1 + \frac{1}{2\pi(1-k^{-2})^2}N^{-2}k^{-2N -2}\Bigl(1 +O(N^{-1})\Bigr)\right),
\quad N \to \infty.
\end{equation}
In the next section this formula will be also proven by the operator technique. 

Let us move now to the bordered Topelitz determinant $D^{B}_{N}[\widehat{\phi}; \widehat{\psi}]$ and consider first the case when
$$
|c_*| <1.
$$
In this case, according to (\ref{start1}), we will only need the asymptotics for $X_{12}(z)$ for $z \in \Omega_0\cup\Omega_1$,
i.e. formula (\ref{X1201}) with $z = c_*$ and $n = N-1$. Moreover, we can use the estimate (\ref{intas2}) for the double integral involved
and get at once that
$$
X_{12}(c_*;N-1) = \alpha(c_*)\left(1 +  \frac{1}{2\pi}
\frac{1}{(k^{-1} - k)(k-c_*)}
k^{-2N+2}N^{-2}\left( 1+ O\left(\frac{1}{N}\right)\right)\right), \quad N \to \infty.
$$
The last equation in conjunction with (\ref{start1}) yields the formula
\begin{equation}\label{DBN100}
D^{B}_{N}[\widehat{\phi}; \widehat{\psi}] =\frac{C_v}{S_v} \alpha(c_*) (1-k^{-2})^{1/4} \left(1 + \frac{1}{2\pi}\frac{1}{1-k^{-2}}\left(-\frac{k}{k-c_*} + \frac{1}{1-k^{-2}}\right)
N^{-2}k^{-2N}\Bigl(1+O(N^{-1})\Bigr)\right).
\end{equation}
This formula, taking into account the definitions (\ref{k-parameter}) and (\ref{hat psi}) of the parameters $k$ and 
$c_*$ and the equation (\ref{AL(c_*)}) for $\alpha(c_*)$ (the case we consider now is $J_v > J_h$) we can rewrite
(\ref{DBN100}) as
\begin{align}\label{magnetization200}
	D_N^B[\widehat{\phi};\widehat{\psi}] &=
	(1-k^{-2})^{1/4}\left(1+\frac{1}{2\pi(1-k^{-2})}\Big(\frac{1}{C_v^2}+\frac{1}{k^2-1}\Big)N^{-2}k^{-2N}\Big(1+O(N^{-1})\Big)
	\right).
\end{align}
This proves Theorem  \ref{th22} for the case $|c_*| < 1$.   The proof of Theorem \ref{th22} for the case $|c_*|>1$ (i.e. if $c_* \in \Om_2$ or $c_*\in \Om_{\infty}$, see Figure \ref{S_contour}) follows from almost identical considerations employed in the case $|c_*|<1$. Let us first discuss the case when $c_* \in \Om_{2}$. In this case we have
\begin{equation}\label{ContrX12Om2}
	\frac{C_v}{S_v} X_{12}(c_*;N-1) = \frac{C_v}{S_v} c^{-N+1}_* \al^{-1}(c_*)\left( \frac{1}{2\pi i} \int_{\Ga_0}a(\tau;N-1)
	\frac{d\tau}{\tau-c_*}+  O\left(\rho^{-3N}\right)\right),
\end{equation}
and
\begin{equation}\label{ContrX11Om2}
\begin{split}
		c^{-N+1}_* \di \frac{C_h}{S_h} X_{11}(c_*;N-1) & = \di \frac{C_h}{S_h} \al(c_*) \left( 1 -\di \frac{1}{4\pi^2}\int_{\Ga_1}b(\mu;N-1)\left[\int_{\Ga_0}a(\tau;N-1)
	\frac{d\tau}{\tau-\mu}\right]\frac{d\mu}{\mu-c_*}+  O(\rho^{-4N})\right) \\
	& -c^{-N+1}_* \di \frac{C_h}{S_h}\alpha^{-1}(c_*)\widehat{\phi}^{-1}(c_*)\left(\frac{1}{2\pi i} \int_{\Ga_0}a(\tau;N-1) 
	\frac{d\tau}{\tau-c_*}+  O(\rho^{-3N})\right).
\end{split}
\end{equation}
These are the terms needed to compute $D^{B}_{N}[\widehat{\phi}; \widehat{\psi}]$ in view of \eqref{start1}. Notice that the contribution from \eqref{ContrX12Om2} cancels the contribution from the second term on the right hand side of \eqref{ContrX11Om2}, as one can simply check that
\begin{equation}
	\frac{C_h}{S_h} \widehat{\phi}^{-1}(c_*) = \frac{C_v}{S_v}.
\end{equation}
Now, from \eqref{AL(c_*)}, \eqref{start1}, \eqref{intas2}, \eqref{DNasRH},  \eqref{ContrX11Om2} we can easily show that \eqref{magnetization2} holds when $c_* \in \Om_2$. Finally we discuss the case $c_* \in \Om_{\infty}$. Equation \eqref{ContrX12Om2} still holds in this case (see \eqref{X122infty}). Using \eqref{AL(c_*)} and \eqref{intas1} we can write
\begin{equation}\label{ContrX12OmInfty}
\frac{C_v}{S_v} X_{12}(c_*;N-1) = \frac{C_vC_h}{S_vS_h} c^{-N+1}_* \left( -\frac{1}{\sqrt{\pi}}\alpha^{2}(k^{-1})\frac{\sqrt{k-k^{-1}}}{k}
\frac{1}{k^{-1} -c_*}k^{-N +1/2}N^{-1/2} \right) \left(1+O\left(\frac{1}{N}\right)\right),
\end{equation}
as $ N \to \infty$. Using \eqref{X11infty} and \eqref{AL(c_*)} we have
\begin{equation}\label{ContrX11OmInfty}
c^{-N+1}_* \di \frac{C_h}{S_h} X_{11}(c_*;N-1) =  1 -\di \frac{1}{4\pi^2}\int_{\Ga_1}b(\mu;N-1)\left[\int_{\Ga_0}a(\tau;N-1)
\frac{d\tau}{\tau-\mu}\right]\frac{d\mu}{\mu-c_*}+  O(\rho^{-4N})
\end{equation}
The asymptotics of the integral on the right hand side can be computed from \eqref{121212}, however, we can not use \eqref{hN3} directly, because when $z \in \Om_{\infty}$ we also get a residue term. To that end by
 a straightforward calculation when $z \in \Om_{\infty}$ we find 
\begin{equation}\label{bintegral-c*inOmInfty}
	\int_{\Gamma_1}b(\mu;n)\frac{d\mu}{(k^{-1} - \mu)(\mu - c_*)} = 2\pi \ic z^{-n-1} \frac{\sqrt{k-z}}{\sqrt{k-z^{-1}}} - 2\ic \frac{\sqrt{k}}{z} \int^{k^{-1}}_0 \frac{t^n \sqrt{k^{-1}-t}}{(t-z^{-1})\sqrt{k-t}\sqrt{t}}\dd t .
\end{equation}
We notice that the residue term (combined with the prefactors coming from \eqref{121212}) exactly cancels out the contribution from \eqref{ContrX12OmInfty}. Finally, the asymptotic expansion of the second term in \eqref{bintegral-c*inOmInfty} can be written as a series involving Beta functions similar to what is shown in equations \eqref{131313} through \eqref{151515}. Finding the asymptotics of the first term in that series using Stirling's formula and then combining this with \eqref{start1}, \eqref{121212}, \eqref{DNasRH} and \eqref{ContrX11OmInfty} finishes the proof of Theorem \ref{th22} for $c_* \in \Om_{\infty}$. 

\section{Asymptotics of Bordered Toeplitz determinants: Operator Theory approach}\label{Section OT}

\subsection{General results}\label{Section OT 1}

For $\phi\in L^1(\T)$ we define the $N\times N$ Toeplitz matrix, 
\begin{equation}\label{toemat}
T_{N}(\phi) := \begin{pmatrix}
\phi_0& \phi_{-1} & \cdots &  \phi_{-N+1}\\
\phi_{1}& \phi_0 & \ddots  & \vdots  \\
\vdots & \ddots & \ddots & \phi_{-1}\\
\phi_{N-1} &  \cdots  & \phi_{1}&\phi_{0}
\end{pmatrix},
\end{equation}
where, as before, $\phi_n$ are the Fourier coefficients of $\phi$.
Occasionally, the notation
$$
[\phi]_n=\frac{1}{2\pi}\int_{0}^{2\pi} \phi(e^{i\theta}) e^{-in\theta}\, d\theta
$$
will be used as well. Clearly, $\det \, T_N(\phi)=D_N[\phi]$.

In what follows, $e_0$ stands for the column vector $(1,0,0,\dots,0)^T$ in $\C^N$, 
and $e_0^T$ signals its transpose, the row vector $(1,0,0,\dots,0)$.

\begin{proposition}\label{p.Bordered.Cramer}
	Let  $\phi,\psi\in L^1(\T)$, and assume that $T_N(\phi)$ is invertible.
	Then
	\begin{equation}\label{Bordered.Cramer}
	D_N^B[\phi;\psi]=
	D_N[\phi]\cdot e_0^T T_N^{-1}(\phi) T_N(\psi) e_0.
	\end{equation}
\end{proposition}
\begin{proof}
Via a rearrangement of rows and columns in \eqref{btd} (or, more formally, by multiplying with a ``flip matrix'' from both sides) we see that
	\begin{equation}\label{btd-2}
	D^{B}_{N}[\phi; \psi] = \det \begin{pmatrix}
	\psi_0& \phi_{-1} & \cdots & \phi_{-N-1} \\
	\psi_{1}& \phi_0 & \ddots  & \vdots \\
	\vdots & \vdots & \ddots &  \phi_{-1} \\
	\psi_{N-1} & \phi_{N-2} & \cdots  &\phi_{0}
	\end{pmatrix}.
	\end{equation}
	By observing that $T_N(\psi)e_0$  is the column vector
	$(\psi_0,\dots,\psi_{N-1})^T$, the statement follows from Cramer's rule.
\end{proof}

In view of computing the asymptotics, the above formula reduces bordered Toeplitz determinants to  
usual Toeplitz determinants $D_N[\phi]$ and the scalar quantity 
\begin{eqnarray}\label{F.N}
F_N[\phi;\psi]:=e_0^T T_N^{-1}(\phi) T_N(\psi) e_0.
\end{eqnarray}
Under appropriate assumptions, the asymptotics of $D_N[\phi]$ is given by the Szeg\H{o}-Widom limit theorem, whereas the asymptotics of the scalar quantity follows from the asymptotics of the inverse of the Toeplitz matrix $T_N(\phi)$. 

We proceed to give some operator-theoretic background; for details we refer to \cite{BS} or \cite{BS1}.
For $\phi\in L^1(\T)$ we define the infinite Toeplitz and Hankel matrices
$$
T(\phi)\cong (\phi_{j-k}),\qquad H(\phi)\cong (\phi_{j+k+1}),\qquad 0\le j,k<\infty.
$$
In case $\phi\in L^\infty(\T)$ these represent bounded linear operators acting on $\ell^2(\ZZ)$.
Note that Toeplitz and Hankel operator satisfy the identity
\begin{eqnarray}\label{Tab}
T(\phi\psi) &=& T(\phi)T(\psi)+H(\phi)H(\tilde{\psi}),
\end{eqnarray}
where $\phi,\psi\in L^\infty(\T)$ and $\tilde{\psi}(z):=\psi(z^{-1})$.
In particular,
\begin{equation}\label{Tabc}
T(\psi_-\phi\psi_+)=T(\psi_-)T(\phi)T(\psi_+)
\end{equation}
if $\psi_\pm\in H^\infty_\pm$, where
$$
H^\infty_\pm:=
\left\{ \, f\in L^\infty(\T)\,:\, f_n=0 \mbox{ for all } \mp n > 0\right\}
$$
are the usual Hardy spaces. We will identify functions in $H^\infty_\pm(\T)$ with their analytic extensions onto
the inside or outside, resp., of $\T$.

Among the various versions of Wiener-Hopf factorization we are going to use the following one.
We say that a function $\phi\in C(\T)$ has a {\em continuous canonical Wiener-Hopf factorization} if it can be written as
$$
\phi(z)=\phi_-(z)\phi_+(z),\qquad|z|=1,
$$
where
$$
\phi_\pm,\phi_\pm^{-1}\in  H^\infty_\pm\cap C(\T).
$$
A sufficient criterium for the existence of such a factorization is that 
$\phi$ belongs to the H\"older class $C^{\eps}(\T)$ for some 
$\eps>0$, is nonvanishing on $\T$ and has winding number zero  (see, e.g., \cite[Sect.~10.2]{BS}).
In this case, the Wiener-Hopf factorization is given by
\begin{equation}\label{eqn.WHfact}
\begin{aligned}
\phi_+(z) &= \exp\left(\sum_{n=0}^\infty z^{n}[\log \phi]_{n}\right),
\quad &
\phi_-(z) & = \exp\left(\sum_{n=1}^\infty z^{-n}[\log \phi]_{-n}\right).
\end{aligned}
\end{equation}
On the other hand, a necessary condition is that $\phi$ is continuous and nonvaninshing on $\T$ and has winding number zero.

For $\phi\in C(\T)$ the Toeplitz operator $T(\phi)$  on $\ell^2(\ZZ)$ is invertible if and only if $\phi$ does not vanish on $\T$ and has winding number zero.
If $\phi$ admits a  continuous canonical Wiener-Hopf factorization then the inverse of $T(\phi)$ is given by
\begin{equation}\label{Tinv}
T^{-1}(\phi) = T(\phi_+^{-1})T(\phi_-^{-1}).
\end{equation}
as can be seen from \eqref{Tabc}.

Finally, let us introduce the finite section projection
$$
P_N:(f_0,f_1,\dots)^T\mapsto (f_0,f_1,\dots, f_{N-1},0,0,\dots)^T
$$
acting on $\ell^2(\ZZ)$. As usual, we will identify $\C^N$ with the image of $P_N$. Correspondingly
we have $P_N T(\phi) P_N= T_N(\phi)$. The complementary projection is
$Q_N=I-P_N$, and we remark that 
$Q_N=V_{N}V_{-N}$ where $V_{N}=T(z^N)$ and $V_{-N}=T(z^{-N})$ are forward and backward shift operators.
Despite having used the notation $e_0$ for finite vectors already, we will also
use it to refer to the infinite column vector,
$$
e_0=(1,0,0,\dots)^T\in\ell^2(\ZZ).
$$
Correspondingly, $e_0^T$ stands for infinite row vector $(1,0,0,\dots)$ or the respective linear functional on $\ell^2(\ZZ)$.

\begin{proposition}\label{p.asymp-1}
	Let $\psi\in L^2(\T)$, 
	and assume that $\phi\in C(\T)$  does not vanish on $\T$ and has winding number zero.
	Then
	$$
	F_N[\phi;\psi] \to F[\phi;\psi]\qquad \mbox{as}\quad N\to\infty,
	$$
	where the constant 
	\begin{equation}\label{eqn.Fconst}
	F[\phi;\psi]:= e_0^T T^{-1}(\phi)T(\psi)e_0.
	\end{equation}
	If, in addition,  $\phi$ has a continuous canonical Wiener-Hopf factorization $\phi=\phi_-\phi_+$, then 
	$$
	F[\phi;\psi]=\frac{[\phi_-^{-1}\psi]_0}{[\phi_+]_0}.
	$$
\end{proposition}
\begin{proof}
	Under above the assumptions on $\phi$, it is well-known (see, e.g., \cite[Sect.~1.5 and 2.3]{BS1})
	that the Toeplitz operator $T(\phi)$ is invertible, that the matrix $T_N(\phi)$ is invertible for sufficiently large $N$, 
	and that $T_N^{-1}(\phi)$ converges to $T^{-1}(\phi)$ strongly on $\ell^2(\ZZ)$ as $N\to\infty$.
	Here we use the afore-mentioned identification of $\C^N$ with a subspace of $\ell^2(\ZZ)$ and the corresponding identification of an $N\times N$ matrix with an operator on $\ell^2(\ZZ)$. Obviously, $T_N(\psi)e_0 \to T(\psi)e_0$ in the norm of $\ell^2(\ZZ)$.
	Therefore, again in the norm
	$$
	T_N^{-1}(\phi) T_N(\psi)e_0 \to T^{-1}(\phi) T(\psi) e_0\qquad \mbox{as} \quad N\to\infty.
	$$
	This proves the first assertion. As to the evaluation of the constant we use \eqref{Tabc} and \eqref{Tinv} 
	to see that 
	$$
	T\iv(\phi) T(\psi)=T(\phi_+^{-1})T(\phi_-^{-1})T(\psi)=T(\phi_+^{-1})T(\phi_-^{-1}\psi)
	$$
	and
	$$
	F[\phi;\psi] = e_0^T T^{-1}(\phi)T(\psi)e_0 = e_0^T T(\phi_+^{-1})T(\phi_-^{-1}\psi)e_0,
	$$
	where in the last expression we interpret  the operators on $\ell^2(\ZZ)$ as infinite matrices.
	Since $T(\phi_+^{-1})$ is lower triangular 
	it follows that $F[\phi;\psi] =[\phi_+^{-1}]_0\cdot [\phi_-^{-1}\psi]_0$.  Observe that $[\phi_+^{-1}]_0=\phi_+^{-1}(0)=1/[\phi_+]_0$.
\end{proof}

The previous results combined with the Szeg\H{o}-Widom limit theorem (\ref{Szego Theorem}) establishes the first order (or leading order) asymptotics for bordered Toeplitz determinants. 
In fact, if we assume that $\phi$ 
is in H\"older class $C^{1+\epsilon}$, does not vanish on the unit circle $\T$, and has winding number zero, then
\begin{equation}
D_N^B[\phi;\psi]=G[\phi]^N E[\phi]\left(F[\phi;\psi]+o(1)\right),\qquad N\to\infty,
\end{equation}
where 
$$
G[\phi]=\exp([\log\phi]_0),\quad 
E[\phi]=\det\, T(\phi)T(\phi^{-1})=\exp\left(\sum_{n\ge 1}n[\log\phi]_n[\log\phi]_{-n}\right).
$$
Thus, what has been stated in Remark \ref{rem.1.3} regarding \eqref{main2} is proved, which is Theorem \ref{thm 1.2} except for the claim
that the error term is decaying exponentially.

Let us emphasize at this point that it can happen that $F[\phi;\psi]$ is zero. 
In this case, the subleading terms in the asymptotics might be of interest as well. 
Later in this section will take up this question.

Let $\phi$ be a function with a continuous canonical Wiener-Hopf factorization $\phi=\phi_-\phi_+$. Each function 
$\psi\in L^2(\T)$ has a unique representation of the form
\begin{equation}\label{RH.dec}
\psi=\phi p_++p_-
\quad \mbox{ with }p_+\in H^2(\T),\quad p_-\in H^2_-(\T),
\end{equation}
where
\begin{align*}
H^2(\T) &=\Big\{f\in L^2(\T)\,:\, f_n=0 \mbox{ for all } n<0\Big\},
\\
H^2_-(\T) &=\Big\{f\in L^2(\T)\,:\, f_n=0 \mbox{ for all } n\ge0\Big\}
\end{align*}
are the corresponding Hardy spaces. Indeed, \eqref{RH.dec} is equivalent to 
\begin{equation}\label{RH1.dec}
\phi_-^{-1}\psi = \phi_+ p_++\phi_-^{-1}p_-,
\end{equation}
from which it can be seen that the terms $p_+$ and $p_-$ are uniquely given by
$$
p_+ =\phi_+^{-1}P[\phi_-^{-1}\psi],\qquad 
p_-= \phi_-(I-P)[\phi_-^{-1}\psi].
$$
Here $P$ is the Riesz projection (i.e., the orthogonal projection on $L^2(\T)$ with range equal to $H^2(\T)$).
We remark that if we consider the Toeplitz operator $T(\phi)$ on $H^2(\T)$ (rather than on $\ell^2(\ZZ)$), then
\begin{align}\label{p+Tiv}
p_+ = T^{-1}(\phi) P[\psi].
\end{align}

\begin{theorem}\label{thm3.3}
Let $\phi\in C(\T)$ 
have a continuous canonical 
Wiener-Hopf factorization $\phi=\phi_-\phi_+$. Assume that  $\psi=\phi p_++p_-$ with $p_+\in H^2(\T)$ and $p_-\in H^2_-(\T)$. Then 
\begin{equation}\label{f.F}
	F[\phi;\psi]=[p_+]_0.
\end{equation}
In particular, 
\begin{align}\label{3.15}
F[\phi; \left( a_{0} + a_1z + \frac{b_0}{z}+\sum_{j=1}^m \frac{b_j z}{z-c_j}\right)\phi] 
&= 
a_{0}+ b_0[\log \phi]_{1} + \sum_{|c_j|<1} b_j\frac{\phi_+(c_j)}{\phi_+(0)},
\\
\label{3.16}
F[\phi; \left(a_0+a_1z+ \frac{b_0}{z} + \sum_{j=1}^m \frac{b_j }{z-c_j}\right)] 
&= 
\frac{a_0}{\phi_+(0)\phi_-(\infty)} - a_1\frac{[\log \phi]_{-1}}{\phi_{+}(0)\phi_-(\infty)}   -\sum_{|c_j|>1} \frac{b_j}{c_j\phi_+(0)\phi_-(c_j)}.
\end{align}
\end{theorem}

\begin{proof}
To prove \eqref{f.F} we note that identity \eqref{RH1.dec} implies that 
	$$
	[ \phi_-^{-1}\psi]_0 = [\phi_+ p_+]_0=[\phi_+]_0[ p_+]_0,
	$$
and thus $F[\phi;\psi]=[p_+]_0$.
We remark that this can also be obtained from \eqref{p+Tiv}.

For the evaluations of $F[\phi;\psi]$ for concrete $\psi$ we  compute the corresponding 
function $p_+(z)$ and then obtain $[p_+]_0=p_+(0)$.
This function can be obtained most conveniently  by writing down the decomposition \eqref{RH1.dec}, 
 $\phi_-^{-1}\psi=\phi_+p_++\phi_-^{-1}p_-$, explicitly.

We start with considering the cases related to \eqref{3.15}.
For $\psi=\frac{z}{z-c}\phi$ with $|c|>1$, we have $p_-=0$, i.e.,
$$
\phi_-^{-1}\psi=
\frac{\phi_+(z)z}{z-c} = \phi_+ p_+,\quad p_+(z)=\frac{z}{z-c},\quad [p_+]_0=p_+(0)=0.
$$
The same conclusion is obtained in the case $\psi=z\phi$. Hence the corresponding terms $a_1$ and $b_j$ (whenever $|c_j|>1$) do not occur on the right hand side of \eqref{3.15}.

For $\psi=\frac{z}{z-c}\phi$ with $|c|<1$, the decomposition is
$$
\phi_-^{-1}\psi=
\frac{\phi_+(z)z}{z-c} = \frac{\phi_+(z)z-\phi_+(c)c}{z-c}+\frac{\phi_+(c)c}{z-c}.
$$
Hence
$$
p_+(z)=\frac{z-\phi_+^{-1}(z)\phi_+(c)c}{z-c},\qquad p_+(0)=\frac{\phi_+(c)}{\phi_+(0)}.
$$
The case $c=0$ covers the case $\psi=\phi$ (related with the coefficient $a_0$) as well.

Lastly, if $\psi=\phi/z$ we observe that 
$$
\phi_-^{-1} \psi =\frac{\phi_+(z)}{z}=\frac{\phi_+(z)-\phi_+(0)}{z}+\frac{\phi_+(0)}{z},
$$
whence 
$$
p_+(z)=\phi_+^{-1}(z)\frac{\phi_+(z)-\phi_+(0)}{z},\qquad [p_+]_0=p_+(0)=\frac{\phi_+'(0)}{\phi_+(0)}=(\log \phi_+(z))'|_{z=0}
=[\log \phi]_1.
$$
Here recall the definition the Wiener-Hopf factor in \eqref{eqn.WHfact}.

Now let us turn to the cases related to \eqref{3.16}.
For $\psi=\frac{1}{z-c}$ with $|c|<1$ or $\psi=\frac{1}{z}$, we will have $p_+=0$ and $p_-=\psi$.
Hence the terms $b_0$ and $b_j$ (whenever $|c_j|<1$) do not occur on the right hand side of \eqref{3.16}.

For $\psi=\frac{1}{z-c}$ with $|c|>1$, the decomposition is 
$$
\phi_-^{-1}\psi=\frac{\phi_-^{-1}(z)}{z-c}=\frac{\phi_-^{-1}(c)}{z-c}+\frac{\phi_-^{-1}(z)-\phi_-^{-1}(c)}{z-c}.
$$
Thus
$$
p_+(z)=\frac{\phi_+^{-1}(z)\phi_-^{-1}(c)}{z-c},\quad\mbox{ and }\quad 
p_+(0)=-\frac{1}{c\phi_+(0)\phi_-(c)}.
$$
The case $\psi=1$ is treated in the same way. Finally, for $\psi=z$ we decompose
$$
\phi_-^{-1}\psi=\eta_+'(0) +(\phi_-^{-1}(z)-\eta_+'(0)\frac{1}{z})z, \quad  \mbox{ with }\quad \eta_+(z)=\phi_-^{-1}(z^{-1}). 
$$
Hence $p_+(z)=\phi_+^{-1}(z) \eta_+'(0)$ and $p_+(0)=\phi_+^{-1}(0)\eta_+'(0)$, where 
$$
\frac{\eta_+'(0)}{\eta_+(0)}=(\log \eta_+(z))'|_{z=0}=-[\log \phi]_{-1}, \quad \mbox{ and }
\quad \eta_+(0)=\phi_-^{-1}(\infty),$$ using 
 \eqref{eqn.WHfact}.
This complete the proof.
\end{proof}

The previous theorem combined with the Szeg\H{o}-Widom Theorem
implies Theorems \ref{main thm}, up to the claim that the error terms are exponentially decaying.
In fact, what has been claimed in Remark \ref{rem.1.3} regarding \eqref{main formula} is proved. 
The identification of the constant $F[\phi;\psi]$ as given in \eqref{3.15} and \eqref{3.16} coincides with the
expression for \eqref{ConstantF2} by taking the relations  \eqref{alpha.phi.pm}  and \eqref{WH.norm.allpah.G}
into account.

\subsection{Higher order asymptotics}

The following ``exact formula'' is the key to the higher order asymptotics of  $F_N[\phi;\psi]$.

\begin{proposition}\label{prop.KN}
	Let $\phi\in C^{1/2+\eps}(\T)$ be a nonvanishing function on the unit circle with winding number zero.
	Assume that $\phi= \phi_-\phi_+$ is its Wiener-Hopf factorization. Let
	$$
	\lambda_N(z)=z^{-N}\lambda(z),\qquad \lambda(z) =\frac{\phi_-(z)}{\phi_+(z)} ,
	$$
	and put $K_N=H(\lambda_N)H(\tilde{\lambda}_N^{-1})$. Then
	\begin{itemize}
		\item[(i)]
		$D_N[\phi]=G[\phi]^NE[\phi]\det(I-K_N)$.
		\item[(ii)]
		$T_N(\phi)$ is invertible if and only if $I-K_N$ is invertible.
		\item[(iii)]
		In this case,
		\begin{equation}\label{TNinv}
		T_N^{-1}(\phi)  =
		T_N(\phi_+^{-1})P_N\left(I- T(\lambda_N^{-1})(I-K_N)^{-1}T(\lambda_N)\right)P_N T_N(\phi_-^{-1}).
		\end{equation}
	\end{itemize}
\end{proposition}
\begin{proof}
	We note that (i) is the Borodin-Okounkov-Case-Geronimo (BOCG) identity (see, e.g., \cite[Sect.~10.40]{BS}),
	and (ii) is an obvious consequence of it. 
	Formula \eqref{TNinv} is basically formula (10.27) or (10.47) in \cite{BS}.
\end{proof}

\begin{theorem}\label{thm.3.1}
	Let $\phi$ have a continuous canonical Wiener-Hopf factorization $\phi=\phi_-\phi_+$, and assume that  $\psi=\phi p_++p_-$ with $p_+\in H^2(\T)$ and $p_-\in H^2_-(\T)$. Then 
\begin{equation}\label{eqn.expan}
	F_N[\phi;\psi]=
	F[\phi;\psi]-\frac{1}{[\phi_+]_0}\cdot  e_0^T T(\lambda^{-1}_N)(I-K_N)^{-1} T(\phi_- z^{-N} p_+)e_0.
\end{equation}
\end{theorem}
\begin{proof}
Note that $T_N(p_-)e_0=0$. Using \eqref{TNinv} we consider
\begin{align*}
	\lefteqn{e_0^T T_N^{-1}(\phi)T_N(\phi p_+)e_0}\qquad\\
	&=
	e_0^T T_N(\phi_+^{-1})P_N\left(I- T(\lambda_N^{-1})(I-K_N)^{-1}T(\lambda_N)\right)P_N T_N(\phi_-^{-1}) T_N(\phi p_+)e_0
	\\
	&= [\phi_+^{-1}]_0\cdot e_0^T
	\left(I- T(\lambda_N^{-1})(I-K_N)^{-1}T(\lambda_N)\right) T(\phi_-^{-1})P_N T(\phi p_+)e_0.
\end{align*}
Apart from the factor $[\phi_+^{-1}]_0=[\phi_+]_0^{-1}$, this decomposes into
\begin{align*}
	\lefteqn{e_0^T T(\phi_-^{-1})P_N T(\phi p_+)e_0
		+e_0^T T(\lambda_N^{-1})(I-K_N)^{-1}T(\lambda_N) T(\phi_-^{-1})Q_N T(\phi p_+)e_0
	}\hspace{2cm}&\\
	&-e_0^T T(\lambda_N^{-1})(I-K_N)^{-1}T(\lambda_N)T(\phi_-^{-1}) T(\phi p_+)e_0.
\end{align*}
The first two terms equal
\begin{align*}
	\lefteqn{e_0^T T(\phi_-^{-1})P_N T(\phi p_+)e_0
		+e_0^T T(\lambda_N^{-1})(I-K_N)^{-1}T(\lambda_N) T(\phi_-^{-1})Q_N T(\phi p_+)e_0}\qquad
	\\
	&=
	e_0^T T(\phi_-^{-1})P_N T(\phi p_+)e_0
	+e_0^T T(\lambda_N^{-1})(I-K_N)^{-1}T(\lambda_N) T(\lambda_N^{-1}) T(\phi_+^{-1})V_{-N}T(\phi p_+)e_0
	\\
	&=
	e_0^T T(\phi_-^{-1})P_N T(\phi p_+)e_0
	+e_0^T T(\lambda_N^{-1}) T(\phi_+^{-1})V_{-N}T(\phi p_+)e_0
	\\
	&=
	e_0^T T(\phi_-^{-1})P_N T(\phi p_+)e_0
	+e_0^T T(\phi_-^{-1}) V_N V_{-N}T(\phi p_+)e_0
	\\
	&=
	e_0^T T(\phi_-^{-1}) T(\phi p_+)e_0=
	e_0^T  T(\phi_+ p_+)e_0 = [\phi_+]_0\cdot [p_+]_0,
\end{align*}
which give the first (constant) term in \eqref{eqn.expan}.
The third term from above equals
\begin{align*}
	\lefteqn{-e_0^T T(\lambda_N^{-1})(I-K_N)^{-1}T(\lambda_N)T(\phi_-^{-1}) T(\phi p_+)e_0}
	\qquad
	\\
	&=
	-e_0^T T(\lambda_N^{-1})(I-K_N)^{-1}T(\lambda_N)T(\phi_+ p_+)e_0
	\\&=
	-e_0^T T(\lambda_N^{-1})(I-K_N)^{-1}T(\phi_- z^{-N} p_+)e_0,
\end{align*}
which provides the second term in \eqref{eqn.expan}.
\end{proof}

The previous result allows to obtain improvements of Proposition \ref{p.asymp-1} by expanding 
the term $(I-K_N)^{-1}$ in formula \eqref{eqn.expan} into the Neumann series.
Notice that $K_N=V_{-N}H(\lambda)H(\tilde{\lambda}^{-1})V_{N}$ and $H(\lambda)H(\tilde{\lambda}^{-1})$ is compact since $\lambda(z)$ is continuous. Hence $K_N$ converges in the operator norm to zero.
In particular, the following conclusions can be drawn.

\begin{corollary}\label{c.32}
	Under the same assumptions as in the previous theorem,
	\begin{align}\label{FN.asym}
	F_N[\psi;\phi] 
	&=
	[p_+]_0
	-\frac{1}{[\phi_+]_0}\cdot  e_0^T T(\lambda^{-1}_N)T(\phi_- z^{-N} p_+)e_0
	\\
	&\qquad
	+O\left(
	\|K_N\|_{L(\ell^2(\ZZ))}\|P[\tilde{\lambda}^{-1}_N]\|_{H^2} \|P[\phi_-z^{-N}p_+]\|_{H^2}
	\right)
	\quad\mbox{ as }N\to\infty.
	\nonumber
	\end{align}
\end{corollary}

Therein, the first term $[p_+]_0$ is the constant, whereas the second one, which can be written as the sum
$$
-\frac{1}{[\phi_+]_0}\sum_{n=N}^\infty 
\left[\frac{\phi_+}{\phi_-}\right]_{-n}\Big[ \phi_-p_+\Big]_{n},
$$
converges to zero as $N\to \infty$.  One should expect that in many cases,  (i.e., unless some ``cancellation'' occurs in the previous sum), the third (or error) term converges faster to zero because it contains $\|K_N\|$.

In the case that the generating functions $\phi$ and $\psi$ are analytic in a neighborhood of $\T$,
exponentially fast convergence can be derived.

\begin{corollary}\label{c.32b}
	Let $\phi(z)$ be analytic an nonvanishing function on the annulus $a_1<|z|<b_1$ with winding number zero,
	and let $\psi(z)$ be analytic on the annulus $a_2<|z|<b_2$, where $a_i<1<b_i$. Then, for each
	$\kappa$ with $\kappa>a_1\max\{b_1^{-1},b_2^{-1}\}$, we have 
	\begin{align}\label{FN.asym2}
	F_N[\phi;\psi] 
	&= F[\phi;\psi]+O(\kappa^N),\qquad N\to \infty,
	\end{align}
and 
	\begin{align}\label{FN.asym3}
	D_N^B[\phi;\psi] 
	&= G[\phi]^NE[\phi]\Big(F[\phi;\psi]+O(\kappa^N)\Big),\qquad N\to \infty.
	\end{align}
\end{corollary}
\begin{proof}
The function  $\lambda(z)^{-1}=\phi_+(z)/\phi_-(z)$ is analytic on $a_1<|z|<b_1$ as well, and hence the Fourier coefficients $[\lambda^{-1}]_{-n}=O(\kappa_1^n)$ as $n\to+\infty$ for each $\kappa_1>a_1$.
As $\phi_-p_+=\phi_-\phi_+^{-1}P[\phi_-^{-1}\psi]$, the function $P[\phi_-^{-1}\psi]$ is analytic on the disc
$|z|<b_2$ and $\phi_-p_+$ is analytic on the annulus $a_1<|z|<\min\{b_1,b_2\}$. Thus, for every 
$\kappa_2>\max\{b_1^{-1},b_2^{-1}\}$, the Fourier coefficients $[\phi_-p_+]_{n}=O(\kappa_2^n)$
as $n\to+\infty$.
Using this information about the Fourier coeffcients, it is easily seen that the second and third term on the right hand side of \eqref{FN.asym} decays as $O(\kappa_1^N\kappa_2^N)$ as $N\to \infty$. This implies \eqref{FN.asym2}. For \eqref{FN.asym3}, we notice that this follows from \eqref{FN.asym2} combined with Proposition\ref{prop.KN}(ii) since a similar estimate can be made for the $\det( I - K_{N})$ term.
\end{proof}

This together with Theorem \ref{thm3.3} completes the proofs of Theorems \ref{main thm} and Theorem \ref{thm 1.2}.

\subsection{Concrete evaluations}

The functions that are of interest in the Ising model are $\phi=\widehat{\phi}$ given by \eqref{Isingphi},
\begin{align}\label{f.phi}
\phi(z)= 
\sqrt{\frac{1-k^{-1}z^{-1}}{1-k^{-1}z}},  \qquad k>1,
\end{align}
and the function $\hat{\psi}$ given by \eqref{hat psi}. Apart from a constant factor, this function can be written as 
$$
\psi(z)=\frac{\phi(z)z-\phi(c)c}{z-c}
$$
with $c=c_*<0$.
Notice that $\phi$ is analytic (and nonzero) on $\C$ except on the branch cut
$$
\Gamma_k:=[0,k^{-1}]\cup [k,+\infty).
$$ 
Therefore, being more general than necessary for the Ising model, we can also allow for complex values $c\notin \Gamma_k$.   Indeed, $\phi(c)$ is well-defined, and therefore $\psi(z)$ is analytic on $\C\setminus \Gamma_k$.

We can apply the formulas established in  Theorem \ref{thm.3.1} and Corollary \ref{c.32} directly to $\psi$, and this is what we will do below.
Alternatively, we could split $\psi$ into two terms
$$
\psi(z)=\phi(z) \frac{z}{z-c} -\frac{\phi(c)c}{z-c}
$$
This basically means that we deal with the functions
$$
\phi(z) \frac{z}{z-c} \quad \mbox{ and }\quad \frac{1}{z-c}.
$$
We do not have to exclude  the values $c\in \Gamma_k$, but to exclude $|c|=1$ and distinguish the cases $|c|>1$ and $|c|<1$.
In the latter case, the asymptotics can be gleaned from Theorem \ref{thm.3.4} and it should be noted that $F_N[\phi;\frac{1}{z-c}]=0$. 
In the former case, the asymptotics of $F_N[\phi;\phi\frac{z}{z-c}]$ is discussed numerically in Section \ref{sec:4:2}, but we refrain from providing the rigorous details.

Note that $\phi$ has Wiener-Hopf factors given by 
\begin{align}\label{f.phi.pm}
\phi_+(z)=(1-k^{-1}z)^{-1/2},\qquad
\phi_-(z)=(1-k^{-1}z^{-1})^{1/2}.
\end{align}
We see that 
\begin{align}\label{f.lambda}
\lambda_N(z)=z^{-N}\lambda(z),\quad \lambda(z)=\sqrt{(1-k^{-1}z^{-1})(1-k^{-1}z)}
\end{align}
We start with the asymptotics of $D_N[\phi]$.

\begin{theorem}
	For $\phi$ given by \eqref{f.phi} with $k>1$, we have that 
	$$
	D_N[\phi]=(1-k^{-2})^{1/4}\left(1+\frac{1}{2\pi(1-k^{-2})^2}N^{-2}k^{-2N-2}(1+O(N^{-1}))\right),\qquad N\to\infty.
	$$
\end{theorem}
\begin{proof}
	We are going to use the BOCG identity stated in Proposition \ref{prop.KN}(i).
	A straightforward evaluation of the constants gives $G[\phi]=1$ and $E[\phi]=(1-k^{-2})^{1/4}$.
	Thus we are left with analyzing
	$$
	\det(I-K_N)=1-\mathrm{trace}\, K_N +O(\|K_N\|_1^2), \qquad N\to\infty.
	$$
	Let us first estimate the trace norm of the operator  
	$K_N=V_{-N}H(\lambda)H(\tilde{\lambda}^{-1})V_N$.
	Since $\lambda(z)$ is analytic on the annulus  $k^{-1}<|z|<k$, the Fourier coefficients decay as
	$[\lambda]_n=O(\kappa^{|n|})$ as $|n|\to\infty$ for each fixed $\kappa>k^{-1}$.
	A straightforward computation of the Hilbert-Schmidt norm of the Hankel operators appearing in $K_N$
	implies that the trace norm of $K_N$ decays exponentially as
	$$
	\|K_N\|_1=O(\kappa^{2N}),\qquad N\to\infty.
	$$
	As a consequence the term $O(\|K_N\|_1^2)$ is negligible in comparison to the other expected terms.
	
	Let us finally compute the asymptotics of the trace of $K_N$.
	Obviously,
	\begin{align*}
	\mathrm{trace}\,  K_N &=\sum_{n=N}^\infty\sum_{j=0}^\infty [\lambda]_{n+j+1}
	[\lambda^{-1}]_{-n-j-1}
	=\sum_{n,j=0}^\infty [\lambda]_{n+j+1+N}
	[\lambda^{-1}]_{-n-j-1-N}
	\\
	&=\sum_{n=0}^\infty (n+1) [\lambda]_{n+1+N}
	[\lambda^{-1}]_{-n-1-N}.
	\end{align*}
	In view of \eqref{f.lambda}, the asymptotics of the Fourier coefficients of $\lambda$ and $\lambda^{-1}$ is given by
	\begin{align}
	[\lambda]_n &= \frac{\sqrt{1-k^{-2}}}{\Gamma(-1/2)}n^{-3/2}k^{-n}\Big(1+O(n^{-1})\Big),
	\\
	[\lambda^{-1}]_{-n}=[\lambda^{-1}]_{n} &=
	\frac{1}{\Gamma(1/2)\sqrt{1-k^{-2}}} n^{-1/2}k^{-n}\Big(1+O(n^{-1})\Big),
	\label{f.3.17}
	\end{align}
	as $n\to\infty$. Here we used Lemma \ref{l.Fc.asym} with $b=k$, $\zeta_0(z)=0$ and
	$\omega=1/2$, $\xi(z)=k^{-1/2}(1-k^{-1}z^{-1})^{1/2}$ in the first case and
	$\omega=-1/2$,  $\xi(z)=k^{1/2}(1-k^{-1}z^{-1})^{-1/2}$ in the second case.
	Hence
	\begin{align*}
	\mathrm{trace}\,  K_N &=
	\sum_{n=0}^\infty \frac{(n+1)}{\Gamma(-1/2)\Gamma(1/2)}(n+N+1)^{-2}k^{-2(N+n+1)}
	\Big(1+O((n+N)^{-1})\Big)
	\\
	&=-\frac{1}{2\pi(1-k^{-2})^2}N^{-2}k^{-2N-2}\Big(1+O(N^{-1})\Big),\qquad N\to\infty,
	\end{align*}
	by Lemma \ref{l.sum.asym2}.
	Combining all this we arrive at
	$$
	\det(I-K_N)=1+\frac{1}{2\pi(1-k^{-2})^2}N^{-2}k^{-2N-2}\Big(1+O(N^{-1})\Big),\qquad N\to\infty,
	$$
	and this proves the assertion.
\end{proof}

Let us now turn to the asymptotics of $F_N[\phi;\psi]$ in a setting which is slightly more general than necessary for the Ising model.

\begin{theorem}\label{thm.3.4}
	For $\phi$ given by \eqref{f.phi} with $k>1$, $c\in\C\setminus[k,+\infty)$, let
	$$
	\psi(z)=\frac{\phi(z)z-d}{z-c}
	$$
	where 
	$$
	d=\begin{cases} \phi(c)c & \mbox{ if } |c|\ge 1
	\\ \mbox{arbitrary} &  \mbox{ if } |c|<1.\end{cases}
	$$
	Then, as $N\to\infty$, 
	\begin{align*}
	F_N[\phi;\psi]
	&=
	\frac{k^{1/2}}{(k-c)^{1/2}}-
	\frac{ck^{1/2}}{2\pi (k-c)^{3/2}(1-k^{-2})} N^{-2} k^{-2N}\Big(1+O(N^{-1})\Big).
	\end{align*}
\end{theorem}
\begin{proof}
	We are going to use Corollary \ref{c.32} and start with identifying the functions therein.
	To compute $p_+$ recall \eqref{RH.dec} and \eqref{RH1.dec} to see that the latter decomposition,
	$\phi_-^{-1}\psi=\phi_+p_++\phi_-^{-1}p_-$ is given by
	$$
	\frac{\phi_+(z)z-\phi_-^{-1}(z)d}{z-c}=\frac{\phi_+(z)z-\phi_+(c)c}{z-c}+\frac{\phi_+(c)c-\phi_-^{-1}(z)d}{z-c}.
	$$
	The first term is analytic for $|z|<k$, while second term is analytic for $|z|>1-\eps$ and vanishes at $z=\infty$.
	Hence
	$$
	p_+(z)=\frac{z-\phi_+^{-1}(z)\phi_+(c)c}{z-c},
	$$
	and
	$$
	[p_+]_0=p_+(0)=\frac{\phi_+(c)}{\phi_+(0)}=\frac{k^{1/2}}{(k-c)^{1/2}}.
	$$
	Furthermore,
	$$
	\phi_-(z)p_+(z)=\phi_-(z)\frac{z-\phi_+^{-1}(z)\phi_+(c)c}{z-c}
	=(1-k^{-1}z^{-1})^{1/2}\frac{z-(1-k^{-1}z)^{1/2}(1-k^{-1}c)^{-1/2}c}{z-c}.
	$$
	Lemma \ref{l.Fc.asym} with $\omega=1/2$, $b=k$, 
	$$
	\xi(z)=-\frac{(1-k^{-1}z^{-1})^{1/2}(k-c)^{-1/2}c}{z-c},\qquad
	\xi(k)=-\frac{(1-k^{-2})^{1/2} c}{(k-c)^{3/2}},
	$$
	gives
	$$
	[\phi_-p_+]_n=-\frac{(1-k^{-2})^{1/2} c}{\Gamma(-1/2)(k-c)^{3/2}} n^{-3/2}k^{-n+1/2}
	\Big(1+O(n^{-1})\Big),\quad n\to\infty.
	$$
	Hence, together with \eqref{f.3.17},
	$$
	[\lambda^{-1}]_{-n}[\phi_-p_+]_n
	=\frac{c}{2\pi (k-c)^{3/2}} n^{-2} k^{-2n+1/2}\Big(1+O(n^{-1})\Big)
	,\quad n\to\infty.
	$$
	Therefore, we get
	\begin{align*}
	e_0^T T(\lambda^{-1}_N)T(\phi_- z^{-N} p_+)e_0
	&=
	\sum_{n=0}^\infty [\lambda^{-1}]_{-n-N}\Big[ \phi_-p_+\Big]_{n+N}
	\\
	&=
	\frac{ck^{1/2}}{2\pi (k-c)^{3/2}(1-k^{-2})} N^{-2} k^{-2N}(1+O(N^{-1}))
	\end{align*}
	using Lemma \ref{l.sum.asym}. Noting that $[\phi_+]_0=\phi_+(0)=1$ and that the error term in 
	Corollary \ref{c.32} decays even as $O(\kappa^{4N})$ (for any fixed $k^{-1}<\kappa<1$), proves the asymptotics.
\end{proof}

\begin{corollary}
	Let $\phi=\widehat{\phi}$ be given by  \eqref{f.phi} with $k=S_hS_v>1$ and
	$$
	\widehat{\psi}(z)=r\frac{\phi(z)z-\phi(c_*)c_*}{z-c_*}
	$$
	with $r=C_v/S_v$ and $c_*=-S_h/S_v$.
	Then, as $N\to\infty$,
	\begin{align}
	\frac{D_N^B[\widehat{\phi};\widehat{\psi}]}{D_N[\widehat{\phi}]} &=1+\frac{1}{2\pi C_v^2 (1-k^{-2})}N^{-2} k^{-2N}\Big(1+O(N^{-1})\Big),
	\\
	D_N^B[\widehat{\phi};\widehat{\psi}] &=
	(1-k^{-2})^{1/4}\left(1+\frac{1}{2\pi(1-k^{-2})}\Big(\frac{1}{C_v^2}+\frac{1}{k^2-1}\Big)N^{-2}k^{-2N}\Big(1+O(N^{-1})\Big)
	\right).
	\end{align}
\end{corollary}
\begin{proof}
	We notice that 
	$$
	\frac{k^{1/2}}{(k-c_*)^{1/2}}= \frac{S_v}{C_v}=\frac{1}{r},
	\qquad 
	\frac{-c_*}{k-c_*}=\frac{1}{C_v^2}.
	$$
	The rest is straightforward computation.
\end{proof}

With this computation we have proved the final two theorems stated in the introduction.


\section{Numerical Verifications}\label{Sec:Num}

In this section, we assume that $\phi \equiv \widehat{\phi}$, the symbol for the Ising model defined by \eqref{Isingphi}.
To fix the problem, we set $\frac{J_h}{k_B}=\frac{1}{2}$ and $\frac{J_v}{k_B}=\frac{1}{4}$ for the $J_h>J_v$ case
and $\frac{J_h}{k_B}=\frac{1}{4}$ and $\frac{J_v}{k_B}=\frac{1}{2}$ for the $J_h<J_v$ case.
Solving  \eqref{Tcr} numerically in both cases, we get $T_c=0.820508964964\cdots$. We thus fix $T=\frac{4}{5}<T_c$  in the following numerical verifications, which ensures that $\widehat{\phi}$ is of Szeg{\H o} type.
Then, we have $k=\sinh(\frac{2 J_h}{k_B T}) \sinh(\frac{2 J_v}{k_B T}) \approx 1.067666675$,
which is, as expected, bigger than $1$. In fact, the reason why we choose $T$ so close to $T_c$ is that the error terms often have factors of the form $N^{-m} k^{-n N}$, so a $k$ slightly greater than $1$ will guarantee the results being not so small for relatively large $N$.

For computing $D_N^B[\phi; \psi]$ and $D_N[\phi]$, from \eqref{phi_n} 
we first compute $\phi_j$, $j=1-N, \cdots, N-1$, and $\psi_j$, $j=0, \cdots N-1$, by the trapezoidal rule up to precision of more than $100$ digits (which is far more than  needed in the following calculations). Then we compute $D_N^B[\phi; \psi]$  and $D_N[\phi]$ directly from \eqref{btd} and \eqref{ToeplitzDet} respectively.

\subsection{Verification of  \eqref{magnetization2}}
Let us define
$$G_N^A:=\left(\frac{D_N^B[\widehat{\phi};\widehat{\psi}]}{\sqrt[4]{1-k^{-2}}} -1\right)\frac{2 \pi (1-k^{-2}) N^2 k^{2 N}}{C_v^{-2}+(k^2-1)^{-1}}.$$
Then formula \eqref{magnetization2} is equivalent to 
\begin{eqnarray}
G_N^A =1+O(N^{-1}). \label{GNAasymp}
\end{eqnarray}
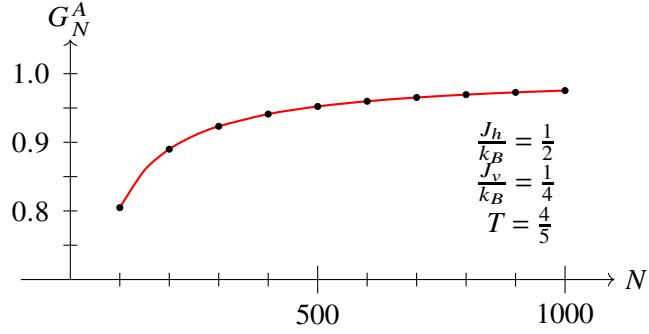
\begin{figure}[h]
	\centering
	\begin{tikzpicture}[scale=1.3]
	\def\xsl{0.005}
	\def\ysl{7}
	\def\FigureDataA{(100*\xsl,0.804766962976850482*\ysl)(150*\xsl, 0.859337718876734363855*\ysl)(200*\xsl,0.8899847129739894345*\ysl)
		(250*\xsl,0.90963999962063914586*\ysl)(300*\xsl,0.923326067623642857*\ysl)(400*\xsl,0.94114215791688982*\ysl)
		(500*\xsl,0.95223496735245265*\ysl)(600*\xsl,0.95980774426959944*\ysl)
		(700*\xsl,0.9653070958461646255*\ysl)(800*\xsl,0.969482260361857755*\ysl)
		(900*\xsl,0.972760221732644385*\ysl)(1000*\xsl,0.9754021736255736016762*\ysl)
	}
	
	\def\FigureDataAA{(100*\xsl,0.804766962976850482*\ysl), (200*\xsl,0.8899847129739894345*\ysl), 
		(300*\xsl,0.923326067623642857*\ysl), (400*\xsl,0.94114215791688982*\ysl),
		(500*\xsl,0.95223496735245265*\ysl), (600*\xsl,0.95980774426959944*\ysl),
		(700*\xsl,0.9653070958461646255*\ysl), (800*\xsl,0.969482260361857755*\ysl),
		(900*\xsl,0.972760221732644385*\ysl), (1000*\xsl,0.9754021736255736016762*\ysl)
	}
	
	\draw[->] (-100*\xsl,0.7*\ysl)--(1100*\xsl,0.7*\ysl) node[right] {$N$};
	\draw[->] (0, 0.7*\ysl)--(0,1.05*\ysl);
	\foreach \x in { 100*\xsl, 200*\xsl, 300*\xsl,400*\xsl, 600*\xsl, 700*\xsl,800*\xsl,900*\xsl } \draw (\x,0.69*\ysl)--(\x,0.71*\ysl);
	\draw (500*\xsl,0.68*\ysl)--(500*\xsl,0.72*\ysl); \draw (1000*\xsl,0.68*\ysl)--(1000*\xsl,0.72*\ysl);
	\foreach \y in {0.75*\ysl,0.8*\ysl, 0.85*\ysl, 0.9*\ysl, 0.95*\ysl, 1.0*\ysl} \draw (2 pt, \y)--(-2pt, \y);
	\node at(500*\xsl,0.65*\ysl){$500$};  \node at(1000*\xsl,0.65*\ysl){$1000$};
	\node at(-80*\xsl,0.8*\ysl){$0.8$}; \node at(-80*\xsl,0.9*\ysl){$0.9$};  \node at(-80*\xsl,1.0*\ysl){$1.0$};
	\draw[thick,color=red] plot[smooth] coordinates {\FigureDataA};
	\foreach \x in \FigureDataAA \draw[thick,fill] \x circle (5*\xsl);
	\node at(0*\xsl,1.08*\ysl){$G_N^A$};
	\node at(900*\xsl,0.9*\ysl){$\frac{J_h}{k_B}=\frac{1}{2}$};
	\node at(900*\xsl,0.84*\ysl){$\frac{J_v}{k_B}=\frac{1}{4}$};
	\node at(910*\xsl,0.78*\ysl){$T=\frac{4}{5}$};
	\end{tikzpicture}
	\caption{Plot of $G_N^A=(\frac{D_N^B[\widehat{\phi};\widehat{\psi}]}{\sqrt[4]{1-k^{-2}}} -1)\frac{2 \pi (1-k^{-2}) N^2 k^{2 N}}{C_v^{-2}+(k^2-1)^{-1}} $.
		Fitting the $10$ points by $\sum\limits_{i=0}^9 g_{-i} N^{-i}$, we get $g_0 \approx 1-5.643811*10^{-11}$, $g_{-1} \approx -25.367279$, $g_{-2} \approx 798.967$, $g_{-3} \approx -30863.1$, $g_{-4}\approx 1.42051 \times 10^6$, $g_5\approx -7.55845\times 10^7$, $g_{-6} \approx 4.41553 \times 10^9 $, $g_{-7}\approx -2.51667\times 10^{11}$, $g_{-8} \approx 1.12018 \times 10^{13}$, $g_{-9}\approx -2.59581\times 10^{14}$. Appending more points such as adding $G_N^A$ for $N=250$, $g_i$ will change 
		with an obvious pattern: the smaller $i$ is, the smaller  percent change is. For example, the change of $g_0$ is
		smaller than $10^{-10}$, and $g_{-6}$ will change to $ 4.55987 \times 10^9$ while $g_{-9}$ will change to $-7.08571\times 10^{14} $.}
	\label{Fig 1}
\end{figure}

Figure \ref{Fig 1} is the numerical result for a case $J_h>J_v$ with $\frac{J_h}{k_B}=\frac{1}{2}$ and $\frac{J_v}{k_B}=\frac{1}{4}$. 
$g_0=1$ and a finite fixed $g_{-1}$ show  asymptotics (\ref{GNAasymp}) is indeed right. A more careful look into the numerical values of $g_i$ suggests that $\sum g_{-i} N^{-i}$ is an asymptotic series.
 
In principle, $G_N^A$ is only defined on integer $N$. The red line is a smooth link of the ten  points obtained by numerical experiments.
The numerical values of $G_N^A$ for other integer $N$ will be visually indistinguishable from the points on the red line.
	\begin{figure}[h]
	\centering
	\begin{tikzpicture}[scale=1.3]
	\def\xsl{0.005}
	\def\ysl{7}
	\def\FigureDataA{(100*\xsl,0.79879711180900302952*\ysl)(150*\xsl,0.854958116774606742*\ysl)(200*\xsl,0.886525934060932323*\ysl)
		(250*\xsl,0.906781868800183664738*\ysl)(300*\xsl,0.920890720636784965297*\ysl)(400*\xsl,0.9392626803698164556*\ysl)
		(500*\xsl,0.9507046930546251*\ysl)(600*\xsl,0.95851721521431047*\ysl)
		(700*\xsl,0.964191353337862185*\ysl)(800*\xsl,0.9684995998681782*\ysl)
		(900*\xsl,0.97188227561001868682*\ysl)(1000*\xsl,0.9746087716916669584854*\ysl)
	}
	
	\def\FigureDataAA{(100*\xsl,0.79879711180900302952*\ysl),(200*\xsl,0.886525934060932323*\ysl),
		(300*\xsl,0.920890720636784965297*\ysl),(400*\xsl,0.9392626803698164556*\ysl),
		(500*\xsl,0.9507046930546251*\ysl),(600*\xsl,0.95851721521431047*\ysl),
		(700*\xsl,0.964191353337862185*\ysl),(800*\xsl,0.9684995998681782*\ysl),
		(900*\xsl,0.97188227561001868682*\ysl),(1000*\xsl,0.9746087716916669584854*\ysl)}
	
	\draw[->] (-100*\xsl,0.7*\ysl)--(1100*\xsl,0.7*\ysl) node[right] {$N$};
	\draw[->] (0, 0.7*\ysl)--(0,1.05*\ysl);
	\foreach \x in { 100*\xsl, 200*\xsl, 300*\xsl,400*\xsl, 600*\xsl, 700*\xsl,800*\xsl,900*\xsl } \draw (\x,0.69*\ysl)--(\x,0.71*\ysl);
	\draw (500*\xsl,0.68*\ysl)--(500*\xsl,0.72*\ysl); \draw (1000*\xsl,0.68*\ysl)--(1000*\xsl,0.72*\ysl);
	\foreach \y in {0.75*\ysl,0.8*\ysl, 0.85*\ysl, 0.9*\ysl, 0.95*\ysl, 1.0*\ysl} \draw (2 pt, \y)--(-2pt, \y);
	\node at(500*\xsl,0.65*\ysl){$500$};  \node at(1000*\xsl,0.65*\ysl){$1000$};
	\node at(-80*\xsl,0.8*\ysl){$0.8$}; \node at(-80*\xsl,0.9*\ysl){$0.9$};  \node at(-80*\xsl,1.0*\ysl){$1.0$};
	\draw[thick,color=red] plot[smooth] coordinates {\FigureDataA};
	\foreach \x in \FigureDataAA \draw[thick,fill] \x circle (5*\xsl);
	\node at(0*\xsl,1.08*\ysl){$G_N^A$};
	\node at(900*\xsl,0.9*\ysl){$\frac{J_h}{k_B}=\frac{1}{4}$};
	\node at(900*\xsl,0.84*\ysl){$\frac{J_v}{k_B}=\frac{1}{2}$};
	\node at(910*\xsl,0.78*\ysl){$T=\frac{4}{5}$}; \label{Fig 2222}
	\end{tikzpicture}
	\caption{Plot of $G_N^A=(\frac{D_N^B[\widehat{\phi};\widehat{\psi}]}{\sqrt[4]{1-k^{-2}}} -1)\frac{2 \pi (1-k^{-2}) N^2 k^{2 N}}{C_v^{-2}+(k^2-1)^{-1}} $.
		Fitting the $10$ points by $\sum\limits_{i=0}^9 g_{-i} N^{-i}$, we get $g_0 \approx 1-5.92566\times10^{-11}$, $g_{-1} \approx -26.191288$, $g_{-2} \approx 830.85573$, $g_{-3} \approx -32206.7$,  $g_{-4} \approx 1.48546\times 10^6$, 
		$g_{-5} \approx-7.91507\times 10^7$, $g_{-6} \approx 4.62828\times 10^9$, $g_{-7} \approx-2.63956 \times 10^{11}$,
		$g_{-8} \approx 1.1753\times 10^{13}$, $g_{-9} \approx-2.72407 \times 10^{14}$.}
	\label{Fig 2}
\end{figure}
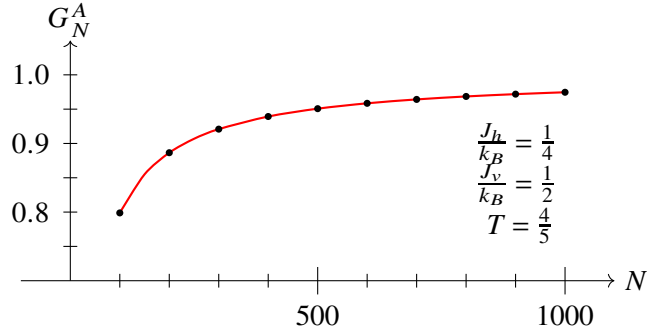

Figure \ref{Fig 2} is the numerical result for a case $J_h<J_v$ with $\frac{J_h}{k_B}=\frac{1}{4}$ and $\frac{J_v}{k_B}=\frac{1}{2}$. Numerical values of  $g_{-i}$ also show the series $\sum g_{-i}N^{-i}$ is an asymptotic one.
$g_0=1$ and a finite fixed $g_{-1}$ show (\ref{GNAasymp}) is also true for this case.

\subsection{The sensitivity for the case  $\psi=\widehat{\phi} \frac{z}{z-c}$ with  $c$ near $1$ }\label{sec:4:2}
If $c<1$, $F_N[\widehat{\phi};\psi]$ is given by Theorem \ref{thm.3.4}. 
Also recall from \eqref{Bordered.Cramer} and \eqref{F.N} that $$F_N[\phi; \psi]=\frac{D_N^B[\phi;\psi]}{D_N[\phi]}.$$
So Theorem \ref{thm.3.4} means that
\begin{equation}
	G_N^B := -\left( \frac{D_N^B[\phi;\psi]}{D_N[\phi]}-\frac{k^{1/2}}{(k-c)^{1/2}} \right)\frac{2 \pi (k-c)^{3/2}(1-k^{-2})}{c k^{1/2}} N^2 k^{2N}=1+O(N^{-1}).
\end{equation} 
Figure 3 is the plot of $G_N^B$ with $c=\frac{975}{1000}$.
$g_0=1$ and a finite fixed $g_{-1}$ verify  Theorem \ref{thm.3.4} numerically.
\begin{figure}[h]
	\centering
	\begin{tikzpicture}[scale=1.3]
	\def\xsl{0.005}
	\def\ysl{7}
	\def\FigureDataA{(100*\xsl,0.78967002252314381*\ysl)(150*\xsl,0.8478673175762918*\ysl)(200*\xsl,0.8807387068223562*\ysl)
		(250*\xsl,0.90189598352968545*\ysl)	(300*\xsl,0.91666415829871734393*\ysl)(400*\xsl,0.935935160083162367*\ysl)
		(500*\xsl,0.94796114535936276*\ysl)(600*\xsl,0.956183373653004*\ysl)
		(700*\xsl,0.9621607919049937*\ysl)(800*\xsl,0.9667025803731339*\ysl)
		(900*\xsl,0.970270627907744944*\ysl)(1000*\xsl,0.973147833033030*\ysl)
	}
	
	\def\FigureDataAA{(100*\xsl,0.78967002252314381*\ysl),(200*\xsl,0.8807387068223562*\ysl),
		(300*\xsl,0.91666415829871734393*\ysl),(400*\xsl,0.935935160083162367*\ysl),
		(500*\xsl,0.94796114535936276*\ysl),(600*\xsl,0.956183373653004*\ysl),
		(700*\xsl,0.9621607919049937*\ysl),(800*\xsl,0.9667025803731339*\ysl),
		(900*\xsl,0.970270627907744944*\ysl),(1000*\xsl,0.973147833033030*\ysl)}
	
	\draw[->] (-100*\xsl,0.7*\ysl)--(1100*\xsl,0.7*\ysl) node[right] {$N$};
	\draw[->] (0, 0.7*\ysl)--(0,1.05*\ysl);
	\foreach \x in { 100*\xsl, 200*\xsl, 300*\xsl,400*\xsl, 600*\xsl, 700*\xsl,800*\xsl,900*\xsl } \draw (\x,0.69*\ysl)--(\x,0.71*\ysl);
	\draw (500*\xsl,0.68*\ysl)--(500*\xsl,0.72*\ysl); \draw (1000*\xsl,0.68*\ysl)--(1000*\xsl,0.72*\ysl);
	\foreach \y in {0.75*\ysl,0.8*\ysl, 0.85*\ysl, 0.9*\ysl, 0.95*\ysl, 1.0*\ysl} \draw (2 pt, \y)--(-2pt, \y);
	\node at(500*\xsl,0.65*\ysl){$500$};  \node at(1000*\xsl,0.65*\ysl){$1000$};
	\node at(-80*\xsl,0.8*\ysl){$0.8$}; \node at(-80*\xsl,0.9*\ysl){$0.9$};  \node at(-80*\xsl,1.0*\ysl){$1.0$};
	\draw[thick,color=red] plot[smooth] coordinates {\FigureDataA};
	\foreach \x in \FigureDataAA \draw[thick,fill] \x circle (5*\xsl);
	\node at(0*\xsl,1.08*\ysl){$G_N^B$};
	\node at(900*\xsl,0.93*\ysl){$\frac{J_h}{k_B}=\frac{1}{2}$};
	\node at(900*\xsl,0.87*\ysl){$\frac{J_v}{k_B}=\frac{1}{4}$};
	\node at(910*\xsl,0.81*\ysl){$T=\frac{4}{5}$};
	\node at(910*\xsl,0.75*\ysl){$c=\frac{975}{1000}$};
	\end{tikzpicture}
	\caption{Plot of $G_N^B =-\left( \frac{D_N^B[\widehat{\phi};\widehat{\psi}]}{D_N[\widehat{\phi}]}-\frac{k^{1/2}}{(k-c)^{1/2}} \right)\frac{2 \pi (k-c)^{3/2}(1-k^{-2})}{c k^{1/2}} N^2 k^{2N}$.
		Fitting the $10$ points by $\sum\limits_{i=0}^9 g_{-i} N^{-i}$, we get $g_0 \approx 1-1.5286\times 10^{-10}$,
		$g_{-1} \approx -27.7534$,  $g_{-2} \approx 938.882$,  $g_{-3} \approx -39575.1$, 
		$g_{-4} \approx 2.02757\times 10^6$,   $g_{-5} \approx -1.22089\times 10^8 $,
		$g_{-6} \approx 8.089\times 10^9 $,   $g_{-7} \approx -5.12614 \times 10^{11} $,
		$g_{-8} \approx 2.45747 \times 10^{13}$,  $g_{-9} \approx -5.96334 \times 10^{14}$.}
	\label{Fig 3}
\end{figure}
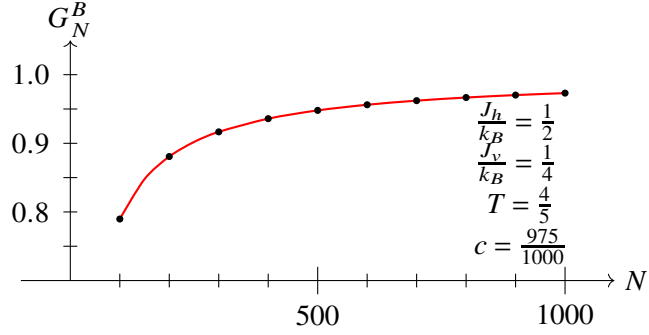

It is not surprising that Figures \ref{Fig 1}, \ref{Fig 2} and \ref{Fig 3} look so similar 
since they all have the same $g_0=1$ and similar $g_{-2}$ and $g_{-3}$.

Now, let us consider the case $c>1$.
For $c>1$, we recall that $p_+(z)=\frac{z}{z-c}$. Now, let us compute $F_N[\psi;\widehat{\phi}]$ by Corollary \ref{c.32}. First, $[p_+]_0=p_+(0)=0$.
Next,  $\widehat{\phi}_- p_+ =\sqrt{1-k^{-1}z^{-1} }\frac{z}{z-c}$.
Therefore,
\begin{eqnarray}
[\widehat{\phi}_- p_+]_n= \int_{\mathbb{T}} \sqrt{1-\frac{1}{k z} } \frac{z^{-n}}{z-c}\frac{dz}{2 \pi i}
=-\sqrt{1-\frac{1}{k c}} c^{-n}, \quad \text{for } n>1. \label{phi-pN}
\end{eqnarray}
Recall that
$$[\lambda^{-1}]_{-n}=\int_{\mathbb{T}} \frac{z^n}{\sqrt{(1-k^{-1}z^{-1}) (1-k^{-1}z)}}\frac{dz}{2 \pi i z}
=\frac{1}{\pi} \int_0^{\frac{1}{k}} \frac{z^{n-1}}{\sqrt{(k^{-1}z^{-1}-1) (1-k^{-1}z)}} dz.$$
We get
\begin{eqnarray}
\sum_{n=N}^\infty -[\lambda^{-1}]_{-n} [\widehat{\phi}_-p_+]_n
&=&\frac{1}{\pi} \sqrt{1-\frac{1}{k c}}\int_0^{\frac{1}{k}} \frac{\sum \limits_{n=N}^\infty z^{n-1} c^{-n}}{\sqrt{(k^{-1}z^{-1}-1) (1-k^{-1}z)}} dz \nonumber\\
&=& \frac{1}{\pi} \sqrt{1-\frac{1}{k c}}\int_0^{\frac{1}{k}} \frac{\left(\frac{z}{c}\right)^{N-1} \frac{1}{c}}{\left(1-\frac{z}{c} \right)\sqrt{(k^{-1}z^{-1}-1) (1-k^{-1}z)}} dz \nonumber\\
&=&\frac{1}{\pi} c^{-N} k^{-N} \sqrt{1-\frac{1}{k c}} \int_0^1 \frac{t^{N-\frac{1}{2}}}{\left(1-\frac{t}{k c} \right) \sqrt{(1-t)(1-k^{-2}t)}} dt \label{sensitivity-1}\\
&\approx&\frac{1}{\sqrt{\pi N}} \frac{c^{-N} k^{-N}}{\sqrt{1-\frac{1}{k c}} \sqrt{1-k^{-2}}} . \label{sensitivity-2}
\end{eqnarray}
(\ref{sensitivity-1}) is the exact value of the second term in Corollary \ref{c.32}. Actually, we do not use  (\ref{sensitivity-2}) since (\ref{sensitivity-1}) itself can be calculated directly.
Let us define $\Delta_N$ as
\begin{equation}
	\Delta_N := \frac{1}{\pi} c^{-N} k^{-N} \sqrt{1-\frac{1}{k c}} \int_0^1 \frac{t^{N-\frac{1}{2}}}{\left(1-\frac{t}{k c} \right) \sqrt{(1-t)(1-k^{-2}t)}} dt-\frac{D_N^B[\widehat{\phi};\psi]}{D_N[\widehat{\phi}]}.
\end{equation}
Then $\Delta_N$ is the negative of the third term in Corollary \ref{c.32}.
\begin{figure}[h]
	\centering
	\begin{tikzpicture}[scale=1.2]
	\def\xsl{0.005}
	\def\ysl{0.01}
	\def\FigureDataA{(100*\xsl,-29.988793853925442959*\ysl)(200*\xsl,-53.69775187059988384*\ysl)
		(300*\xsl,-76.7738681628204375556*\ysl)(400*\xsl,-99.579329829395673*\ysl)
		(500*\xsl,-122.2333924108821*\ysl)(600*\xsl,-144.790502755500954*\ysl)
		(700*\xsl,-167.2801517742679508*\ysl)(800*\xsl,-189.7201289185452199*\ysl)
		(900*\xsl,-212.121997560115069*\ysl)(1000*\xsl,-234.4936999058024675655*\ysl)
	}
	
	\def\FigureDataAA{(100*\xsl,-29.988793853925442959*\ysl),(200*\xsl,-53.69775187059988384*\ysl),
		(300*\xsl,-76.7738681628204375556*\ysl),(400*\xsl,-99.579329829395673*\ysl),
		(500*\xsl,-122.2333924108821*\ysl),(600*\xsl,-144.790502755500954*\ysl),
		(700*\xsl,-167.2801517742679508*\ysl),(800*\xsl,-189.7201289185452199*\ysl),
		(900*\xsl,-212.121997560115069*\ysl),(1000*\xsl,-234.4936999058024675655*\ysl)
	}
	
	\draw[->] (-100*\xsl,0*\ysl)--(1100*\xsl,0*\ysl) node[right] {$N$};
	\draw[->] (0, -250*\ysl)--(0,50*\ysl);
	\foreach \x in { 100*\xsl, 200*\xsl, 300*\xsl,400*\xsl, 600*\xsl, 700*\xsl,800*\xsl,900*\xsl } \draw (\x,-5*\ysl)--(\x,5*\ysl);
	\draw (500*\xsl,-10*\ysl)--(500*\xsl,10*\ysl); \draw (1000*\xsl,-10*\ysl)--(1000*\xsl,10*\ysl);
	\foreach \y in {-50*\ysl,-100*\ysl, -150*\ysl, -200*\ysl} \draw (2 pt, \y)--(-2pt, \y);
	\node at(500*\xsl,30*\ysl){$500$};  \node at(1000*\xsl,30*\ysl){$1000$};
	\node at(-80*\xsl,-100*\ysl){$-100$}; \node at(-80*\xsl,-200*\ysl){$-200$}; 
	\draw[thick,color=red] plot[smooth] coordinates {\FigureDataA};
	\foreach \x in \FigureDataAA \draw[thick,fill] \x circle (5*\xsl);
	\node at(0*\xsl,75*\ysl){$\ln \Delta_N$};
	\node at(950*\xsl,-40*\ysl){$\frac{J_h}{k_B}=\frac{1}{2}$};
	\node at(950*\xsl,-80*\ysl){$\frac{J_v}{k_B}=\frac{1}{4}$};
	\node at(960*\xsl,-120*\ysl){$T=\frac{4}{5}$};
	\node at(940*\xsl,-160*\ysl){$c=\frac{1025}{1000}$};
	\end{tikzpicture}
	\caption{Plot of $\ln \Delta_N$, where
		$ \Delta_N= \frac{1}{\pi} c^{-N} k^{-N} \sqrt{1-\frac{1}{k c}} \int_0^1 \frac{t^{N-\frac{1}{2}}}{\left(1-\frac{t}{k c} \right) \sqrt{(1-t)(1-k^{-2}t)}} dt-\frac{D_N^B[\widehat{\phi};\psi]}{D_N[\widehat{\phi}]}$.
		Fitting the $10$ points by $g_1 N+g_L \ln N +\sum \limits_{i=0}^7 g_{-i} N^{-i} $, we get $g_1 \approx -0.2211193827735$,  $g_L \approx -2.499999847$ and $g_0 \approx 3.928447182$.
		Notice that the numerical $g_1$ and $g_L$ obtained by the fitting of the $10$ points are very close to $3 \ln k^{-1}+ \ln c^{-1} \approx -0.22111938274235$ and $-\frac{5}{2}$. This suggests that $g_1=3 \ln k^{-1}+ \ln c^{-1}$ and $g_L=-\frac{5}{2}$ are exact.}
	\label{Fig 4}
\end{figure}
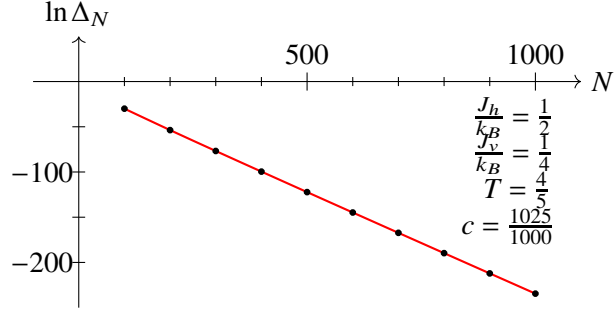

Figure \ref{Fig 4} is the plot of $\Delta_N$ with $c=\frac{1025}{1000}$.
The numerical results mean 
\begin{equation}
	\frac{D_N^B[\widehat{\phi};\psi]}{D_N[\widehat{\phi}]}=\frac{1}{\pi} c^{-N} k^{-N} \sqrt{1-\frac{1}{k c}} \int_0^1 \frac{t^{N-\frac{1}{2}}}{\left(1-\frac{t}{k c} \right) \sqrt{(1-t)(1-k^{-2}t)}} dt+O(N^{-\frac{5}{2}} c^{-N} k^{-3 N}),
\end{equation}
in this case.


\section{Appendices}\label{Appendices}

\subsection{Solution of the Riemann-Hilbert problem for BOPUC with Szeg{\H o}-type symbols}\label{Appendix Y-RHP}

The following Riemann-Hilbert problem for BOPUC is due to J.Baik, P.Deift and K.Johansson.

\begin{itemize}
\item  \textbf{RH-X1} \qquad $X:\C\setminus \T \to \C^{2\times2}$ is analytic,
\item \textbf{RH-X2} \qquad  The limits of $X(\ze)$ as $\ze$ tends to $z \in \T $ from the inside and outside of the unit circle exist, and are denoted $X_{\pm}(z)$ respectively and are related by

\begin{equation}
X_+(z)=X_-(z)\begin{pmatrix}
1 & z^{-n}\phi(z) \\
0 & 1
\end{pmatrix}, \qquad  z \in \T,
\end{equation}

\item \textbf{RH-X3} \qquad  As $z \to \infty$

\begin{equation}
X(z)=\big( I + O(z^{-1}) \big) z^{n \sigma_3},    
\end{equation}
\end{itemize}

(see \cite{D},\cite{DIK},\cite{CIK}). Below we show the standard steepest descent analysis to asymptotically solve this problem, in the case where $\phi$ is a symbol analytic in a neighborhood of the unit circle and with zero winding number. Note that the symbol $\phi$ associated to the 2D Ising model in the low temperature regime enjoys these properties.  We first normalize the behavior at $\infty$ by defining \begin{equation}\label{40}
T(z;n) := \begin{cases}
X(z;n)z^{-n\sigma_3}, & |z|>1, \\
X(z;n), & |z|<1.
\end{cases}
\end{equation}
The function $T$ defined above satisfies the following RH problem

\begin{itemize}
\item  \textbf{RH-T1} \qquad $T(\cdot;n) :\C\setminus \T \to \C^{2\times2}$ is analytic,
\item  \textbf{RH-T2} \qquad    $T_{+}(z;n)=T_{-}(z;n)\begin{pmatrix}
z^n & \phi(z) \\
0 & z^{-n}
\end{pmatrix}, \qquad  z \in \T$,
\item  \textbf{RH-T3} \qquad   $T(z;n)=I+O(1/z), \qquad z \to \infty,$
\end{itemize}
So $T$ has a highly-oscillatory jump matrix as $n \to \infty$. The next transformation yields a Riemann Hilbert problem, normalized at infinity, having an exponentially decaying jump matrix on the \textit{lenses}. Note that we have the following factorization of the jump matrix of the $T$-RHP: \begin{equation}\label{Appendix Pure Toeplitz G_T factorization}
\begin{pmatrix}
z^n & \phi(z) \\
0 & z^{-n}
\end{pmatrix} = \begin{pmatrix}
1 & 0 \\
z^{-n}\phi(z)^{-1} & 1
\end{pmatrix}\begin{pmatrix}
0 & \phi(z) \\
-\phi(z)^{-1} & 0
\end{pmatrix}\begin{pmatrix}
1 & 0 \\
z^{n}\phi(z)^{-1} & 1
\end{pmatrix}  \equiv J_{0}(z;n)J^{(\infty)}(z)J_{1}(z;n).
\end{equation}
Now, we define the following function : \begin{equation}\label{42}
S(z;n):=\begin{cases}
T(z;n)J^{-1}_{1}(z;n), & z \in \Om_1, \\
T(z;n)J_{0}(z;n), & z \in \Om_2, \\
T(z;n), & z\in \Om_0\cup \Om_{\infty}.
\end{cases}
\end{equation}
Also introduce the following function on $\Ga_S := \Ga_0 \cup \Ga_1 \cup \T$  \begin{equation}\label{43}
J_{ S}(z;n)=\begin{cases}
J_{1}(z;n), & z \in \Ga_0, \\
J^{(\infty)}(z), & z \in \T, \\
J_{0}(z;n), & z \in \Ga_1. \\
\end{cases}
\end{equation} 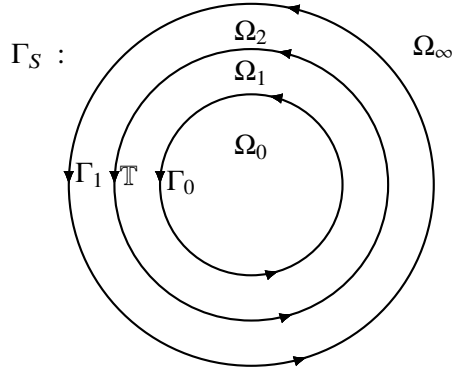
\begin{figure}
\centering

\begin{tikzpicture}[scale=0.4]

\draw[ ->-=0.22,->-=0.5,->-=0.8,thick] (-3,0) circle (4.5cm);

\draw[ ->-=0.22,->-=0.5,->-=0.8,thick] (-3,0) circle (3cm);

\draw[ ->-=0.22,->-=0.5,->-=0.8,thick] (-3,0) circle (6cm);



\node at (-3,4.45) [below] {$\Om_1$};

\node at (-3,5.8) [below] {$\Om_2$};

\node at (-8.3,-0.4) [above] {$\Ga_1$};



\node at (-6.4,0.3) [left] {$\T$};

\node at (-5.3,0.8) [below] {$\Ga_0$};

\node at (-3,2) [below] {$\Om_0$};



\node at (3,5.3) [below] {$\Om_{\infty}$};

\node at (-10,5) [below] {$\Ga_S \ :$};
\end{tikzpicture}

\caption{Opening of lenses: the jump contour for the $S$-RHP.}
\label{S_contour}
\end{figure}
We have the following Riemann-Hilbert problem for $S(z;n)$

\begin{itemize}
\item \textbf{RH-S1} \qquad  $ S(\cdot;n):\C\setminus \Ga_S \to \C^{2\times2}$ is analytic.
\item \textbf{RH-S2} \qquad  $S_{+}(z;n)=S_{-}(z;n) J_{S}(z;n), \qquad  z \in \Ga_S.$
\item \textbf{RH-S3} \qquad  $ S(z;n)=I+O(1/z), \qquad \text{as} \ z \to \infty$.
\end{itemize}
Note that the matrices $J_0(z;n)$ and $J_1(z;n)$ tend to the identity matrix uniformly on their respective contours, exponentially fast as $n \to \infty$.
\subsubsection{Global parametrix RHP}

We are looking for a piecewise analytic function $P^{(\infty)}(z): \C \setminus \T : \to \C^{2\times2}$ such that \begin{itemize}
\item \textbf{RH-Global1} \qquad    $P^{(\infty)}$ is holomorphic in $\C \setminus \T$.

\item \textbf{RH-Global2} \qquad for $z\in \T$ we have \begin{equation}\label{44}
P_+^{(\infty)}(z)=P_-^{(\infty)}(z)  \begin{pmatrix}
0 & \phi(z) \\
- \phi^{-1}(z) & 0
\end{pmatrix}.
\end{equation}

\item \textbf{RH-Global3} \qquad   $P^{(\infty)}(z)=I+O(1/z),  \qquad  \text{as} \ z \to \infty$.

\end{itemize}
We can find a piecewise analytic function $\al$ which solves the following scalar multiplicative Riemann-Hilbert problem \begin{equation}\label{47}
\al_+(z)=\al_-(z)\phi(z) \qquad  z \in \T.
\end{equation}
By Plemelj-Sokhotski formula we have \begin{equation}\label{45}
\al(z)=\exp \left[ \frac{1}{2 \pi i } \int_{\T} \frac{\ln(\phi(\tau))}{\tau-z}d\tau \right],
\end{equation}
Now, using (\ref{47}) we have the following factorization \begin{equation}\label{46}
\begin{pmatrix}
0 & \phi(z) \\
- \phi^{-1}(z) & 0
\end{pmatrix}= \begin{pmatrix}
\al_-^{-1}(z) & 0 \\
0 & \al_-(z)
\end{pmatrix} \begin{pmatrix}
0 & 1 \\
- 1 & 0
\end{pmatrix} \begin{pmatrix}
\al^{-1}_+(z) & 0 \\
0 & \al_+(z)
\end{pmatrix}.
\end{equation}
So, the function \begin{equation}\label{48}
P^{(\infty)}(z) :=\begin{cases}
\begin{pmatrix}
0 & \al(z) \\
-\al^{-1}(z) & 0
\end{pmatrix},  & |z|<1, \\
\begin{pmatrix}
\al(z) & 0 \\
0 & \al^{-1}(z)
\end{pmatrix}, & |z|>1,
\end{cases}
\end{equation}
satisfies (\ref{44}). Also, by the properties of the Cauchy integral, $P^{(\infty)}(z)$ is holomorphic in $\C \setminus \T$. Moreover, $\al(z)=1+O(z^{-1})$, as $z\to \infty$ and hence \begin{equation}\label{49}
P^{(\infty)}(z)=I+O(1/z), \qquad  z \to \infty.
\end{equation}
Therefore $P^{(\infty)}$ given by (\ref{48}) is the unique solution of the global parametrix Riemann-Hilbert problem.
\subsubsection{Small-norm RHP}

Let us consider the ratio

\begin{equation}
R(z;n):= S(z;n) \left[ P^{(\infty)}(z) \right]^{-1}.
\end{equation}
We have the following Riemann-Hilbert problem for $R(z;n)$

\begin{itemize}
\item \textbf{RH-R1} \qquad $ R$ is holomorphic in $\C\setminus (\Ga_0 \cup \Ga_1)$.
\item \textbf{RH-R2} \qquad    $R_{+}(z;n)=R_{-}(z;n) J_{R}(z;n), \qquad  z \in \Ga_0 \cup \Ga_1 =: \Sigma_R$,
\item \textbf{RH-R3} \qquad $ R(z;n)=I+O(1/z) \qquad \text{as} \ z \to \infty$.
\end{itemize}
This Riemann Hilbert problem is solvable for large $n$ (\cite{DKM+99a},\cite{DKM+99b})  and $R(z;n)$ can be written as
\begin{equation}\label{Appendix pure toep R series}
R(z;n) = I + R_1(z;n) + R_2(z;n) + R_3(z;n) + \cdots, \ \ \ \  \ \ n \geq n_0
\end{equation}
where $R_k$ can be found recursively. Indeed 
\begin{equation}\label{Recursive R_k Def}
		R_k(z;n) = \frac{1}{2\pi \ic}\int_{\Sigma_R} \frac{\left[ R_{k-1}(\mu;n)\right]_-  \left( J_R(\mu;n)-I\right) }{\mu-z}\dd \mu,  \qquad  z \in \C \setminus \Sigma_R, \qquad k\geq1.
		\end{equation}	
It is easy to check that $R_{2\ell}(z;n)$ is diagonal and $R_{2\ell+1}(z;n)$ is off-diagonal; $\ell \in \N \cup \{0\}$, and that 
\begin{equation}\label{R_k's are small}
R_{k,ij}(z;n) = \frac{O(\rho^{-kn})}{1+|z|}, \qquad n \to \infty, \qquad k \geq 1, \qquad z\in \C \setminus \Sigma_R,
\end{equation}
where $\rho$ (resp. $\rho^{-1}$) is the radius of $\Ga_1$(resp. $\Ga_0$). Let us compute $R_1(z;n)$; we have
\begin{equation}\label{JR-I}
J_R(z)-I = \begin{cases}
P^{(\infty)}(z) \begin{pmatrix}
0 & 0 \\
z^n \phi^{-1}(z) & 0 \end{pmatrix} \left[ P^{(\infty)}(z) \right]^{-1}, & z \in \Ga_0, \\
P^{(\infty)}(z) \begin{pmatrix}
0 & 0 \\
z^{-n} \phi^{-1}(z) & 0 \end{pmatrix} \left[ P^{(\infty)}(z) \right]^{-1}, & z \in \Ga_1,
\end{cases} =  \begin{cases}
\begin{pmatrix}
0 & -z^{n} \phi^{-1}(z)\al^{2}(z) \\
0 & 0 \end{pmatrix}, & z \in \Ga_0, \\
\begin{pmatrix}
0 & 0 \\
z^{-n} \phi^{-1}(z)\al^{-2}(z) & 0 \end{pmatrix},  & z \in \Ga_1.
\end{cases}
\end{equation}
Therefore

\begin{equation}\label{R1}
R_1(z;n)=  \begin{pmatrix}
0 & -\di \frac{1}{2\pi i}\int_{\Ga_0} \frac{\tau^{n} \phi^{-1}(\tau)\al^{2}(\tau)}{\tau-z}d\tau \\ \di \frac{1}{2\pi i}\int_{\Ga_1} \frac{\tau^{-n} \phi^{-1}(\tau)\al^{-2}(\tau)}{\tau-z}d\tau
& 0 \end{pmatrix}.
\end{equation}
\subsubsection{Tracing back Riemann-Hilbert transformations}
If we trace back the Riemann-Hilbert problems $R \mapsto S \mapsto T \mapsto Y $ we will obtain

\begin{equation}\label{X in terms of R exact}
X(z;n) = R(z;n)\begin{cases} \begin{pmatrix} \al(z) & 0 \\ 0 & \al^{-1}(z) \end{pmatrix}z^{n\sigma_3}, & z \in \Om_{\infty}, \\
\begin{pmatrix} \al(z) & 0 \\ -z^{-n}\al^{-1}(z) \phi^{-1}(z) & \al^{-1}(z) \end{pmatrix}z^{n\sigma_3}, & z \in \Om_{2}, \\
\begin{pmatrix} z^{n}\al(z) \phi^{-1}(z) & \al(z) \\ -\al^{-1}(z) & 0 \end{pmatrix}, & z \in \Om_{1}, \\  \begin{pmatrix} 0 & \al(z) \\ -\al^{-1}(z) & 0 \end{pmatrix}, & z \in \Om_{0}, \\  
\end{cases}
\end{equation}
where for $z \in \C \setminus \Sigma_{R}$, as $n \to\infty$, we have
\begin{equation}\label{R asymp}
R(z;n)=\di \begin{pmatrix}
1 + \frac{ O(\rho^{-2n})}{ 1+|z| } & R_{1,12}(z;n)+ \frac{ O(\rho^{-3n})}{ 1+|z| } \\
R_{1,21}(z;n)+ \frac{ O(\rho^{-3n})}{ 1+|z| } & 1 +\frac{ O(\rho^{-2n})}{ 1+|z| }
\end{pmatrix}.    
\end{equation}

\subsection{Proof of Theorem \ref{thm 1.2} using the Riemann-Hilbert approach}\label{RHP proof thm 1.2} 
For the inverse of a Toeplitz matrix $T_n[\phi]=\{\phi_{j-k}\}^{n-1}_{j,k=0}$, we have 

\begin{equation}\label{InverseToeplitz Resolvent}
\left(T^{-1}_n[\phi] \right)_{j,k}=\de_{jk}+ \langle \boldsymbol{\mathcal{R}^{(\phi)}_n} [z^k],z^j \rangle, \qquad 0 \leq j,k \leq n-1,
\end{equation}
where $\de_{jk}$ is the Kronecker delta function,
\begin{equation}
\langle f(z),g(z) \rangle = \int_{\T} f(z) \overline{g(z)} \frac{\dd z}{2\pi \ic z},
\end{equation}
and
\begin{equation}\label{resolvent operator}
\boldsymbol{\mathcal{R}^{(\phi)}_n}: f(z) \mapsto \int_{\T} \mathcal{R}^{(\phi)}_n(z,w) f(w) \dd w 
\end{equation}
is the \textit{Resolvent} operator with the kernel
\begin{equation}\label{resolvent kernel}
\mathcal{R}^{(\phi)}_n(z,w) = \frac{X^{(\phi)}_{11}(z)X^{(\phi)}_{21}(w)-X^{(\phi)}_{11}(w)X^{(\phi)}_{21}(z)}{z-w} \frac{\phi(w)-1}{2\pi \ic w^n},
\end{equation}
where $X^{(\phi)}_{11}(z) \equiv X^{(\phi)}_{11}(z;n)$ and $X^{(\phi)}_{12}(z)\equiv X^{(\phi)}_{12}(z;n)$ are the entries of the solution to the \textbf{RH-X1} through \textbf{RH-X3}. In terms of the associated biorthogonal polynomials, in view of \eqref{Toeplitz-OP-solution}, we can write
\begin{equation}\label{resolvent kernel in terms of OPs}
\mathcal{R}^{(\phi)}_n(z,w) =\frac{\sqrt{D_{n-1}[\phi]D_{n+1}[\phi]}}{D_{n}[\phi]} \frac{Q_n(w)Q^*_{n-1}(z)-Q_n(z)Q^*_{n-1}(w)}{z-w} \frac{\phi(w)-1}{2\pi \ic w^n},
\end{equation}
where we have used the standard notation
\[ P^{*}_n(z):= z^n P_n(z^{-1})  \]
for a polynomial $P_n(z)$ of degree $n$. 

Let $\Vec{x} = (x_0, x_1, \cdots, x_{N-1})^T$ and $\Vec{\psi} = (\psi_{N-1}, \psi_{N-2}, \cdots, \psi_{0})^T$. Applying the Cramer's rule to the linear system $T_n[\tilde{\phi}] \Vec{x} = \Vec{\psi}$ gives
\begin{equation*}
x_{N-1} = \frac{\det \begin{pmatrix}
	\Tilde{\phi}_0 & \Tilde{\phi}_{-1} & \cdots  & \Tilde{\phi}_{-N+2}  & \psi_{N-1} \\
	\Tilde{\phi}_1 & \Tilde{\phi}_{0} & \cdots  & \Tilde{\phi}_{-N+3} & \psi_{N-2} \\
	\vdots & \vdots & \cdots  & \vdots & \vdots \\
	\Tilde{\phi}_{N-1} & \Tilde{\phi}_{N-2} & \cdots & \Tilde{\phi}_{1} & \psi_0 \end{pmatrix}}{D_N[\tilde{\phi}]} = \frac{\det \begin{pmatrix}
	\phi_0 & \phi_{1} & \cdots  & \phi_{N-2}  & \psi_{N-1} \\
	\phi_{-1} & \phi_{0} & \cdots  & \phi_{N-3} & \psi_{N-2} \\
	\vdots & \vdots & \cdots  & \vdots & \vdots \\
	\phi_{1-N} & \phi_{2-N} & \cdots & \phi_{-1} & \psi_0 \end{pmatrix}}{D_N[\phi]}.
\end{equation*}
Comparing this with \eqref{btd} we observe that
\begin{equation}\label{bdt t xN-1}
D^B_N[\phi;\psi] = D_{N}[\phi] x_{N-1} 
\end{equation}
In view of (??), \eqref{InverseToeplitz Resolvent}, and \eqref{bdt t xN-1} we have
\begin{equation*}
F_N[\phi;\psi]\equiv x_{N-1}= \sum_{\ell=0}^{N-1} \left(T^{-1}_N[\tilde{\phi}] \right)_{N-1,\ell} \psi_{N-1-\ell} = \sum_{\ell=0}^{N-1} \left( \de_{N-1,\ell} + \left\langle \boldsymbol{\mathcal{R}^{(\tilde{\phi})}_N} [z^{\ell}],z^{N-1} \right\rangle \right) \psi_{N-1-\ell}  .
\end{equation*}
Thus
\begin{equation}\label{F_N phi psi}
F_N[\phi;\psi] = \psi_0 + \sum_{\ell=0}^{N-1}  \left\langle \boldsymbol{\mathcal{R}^{(\tilde{\phi})}_N} [z^{\ell}],z^{N-1} \right\rangle \psi_{N-1-\ell}
\end{equation}
From \eqref{resolvent operator} and \eqref{resolvent kernel} we have
\begin{equation}
\begin{split}
& \mathcal{I}_N := \sum_{\ell=0}^{N-1}  \left\langle \boldsymbol{\mathcal{R}^{(\tilde{\phi})}_N} [z^{\ell}],z^{N-1} \right\rangle \psi_{N-1-\ell} \\ & \sum_{\ell=0}^{N-1} \left\{\int_{\T} \left( \int_{\T} \frac{X^{(\tilde{\phi})}_{11}(z)X^{(\tilde{\phi})}_{21}(w)-X^{(\tilde{\phi})}_{11}(w)X^{(\tilde{\phi})}_{21}(z)}{z-w} \frac{\tilde{\phi}(w)-1}{2\pi \ic w^N} w^{\ell} \dd w \right) z^{-N} \frac{\dd z}{2\pi \ic} \right\}\psi_{N-1-\ell}.
\end{split}
\end{equation}
Note that
\begin{equation}
\sum^{N-1}_{\ell=0} w^{\ell-N} \psi_{N-1-\ell} = \frac{1}{w} \sum^{N-1}_{k=0} w^{-k} \psi_{k} = \frac{1}{w} \psi_{i}  \left(\frac{1}{w}\right)+\mathcal{O}\left(e^{-c_0N}\right), 
\end{equation}
where $c_0$ is some positive constant and  \begin{equation}
\psi(z) = \sum^{\infty}_{k=0} \psi_k z^k + \sum^{\infty}_{k=1} \psi_{-k} z^{-k} \equiv \psi_{i}(z)+\psi_{o}(z).
\end{equation}
Therefore 

\begin{equation}
\mathcal{I}_N \simeq \int_{\T}  \int_{\T} \frac{X^{(\tilde{\phi})}_{11}(z)X^{(\tilde{\phi})}_{21}(w)-X^{(\tilde{\phi})}_{11}(w)X^{(\tilde{\phi})}_{21}(z)}{z-w} \frac{\tilde{\phi}(w)-1}{2\pi \ic w} \psi_{i}  \left(\frac{1}{w}\right)  z^{-N} \frac{\dd w  \dd z}{2\pi \ic}
\end{equation}
Let $$\mathcal{D}(z):= \exp\left[ \frac{1}{2\pi \ic} \int_{\T} \frac{\ln(\tilde{\phi}(\tau))}{\tau-z}\dd \tau  \right].$$
One can easily check that \begin{equation}\label{D to al}
\mathcal{D}(z) = \frac{\al(0)}{\tilde{\al}(z)},
\end{equation}
where $\al$ is the Szeg{\H o} function corresponding to the symbol $\phi$, given by \eqref{45}. For $\mathcal{D}$ we have 
\begin{equation}\label{D jump condition}
\mathcal{D}_+(z)=\mathcal{D}_-(z) \tilde{\phi}(z), \qquad z\in \T.
\end{equation}
Recall from \eqref{R_k's are small}, \eqref{X in terms of R exact}, and \eqref{R asymp} that
\begin{equation}
X^{(\tilde{\phi})}_{11}(z) \simeq \mathcal{D}_+(z) z^N \tilde{\phi}^{-1}(z), \qandq X^{(\tilde{\phi})}_{11}(z) \simeq - \mathcal{D}^{-1}_+(z).
\end{equation}
Therefore
\begin{equation}
\begin{split}
\mathcal{I}_N & \simeq \int_{\T}  \int_{\T} \frac{- \mathcal{D}^{-1}_+(w)\mathcal{D}_+(z) z^N \tilde{\phi}^{-1}(z)+\mathcal{D}^{-1}_+(z)\mathcal{D}_+(w) w^N \tilde{\phi}^{-1}(w)}{z-w} \frac{\tilde{\phi}(w)-1}{2\pi \ic w} \psi_{i}  \left(\frac{1}{w}\right)  z^{-N} \frac{\dd w  \dd z}{2\pi \ic} \\ & = \int_{\T}  \int_{\T} \frac{- \mathcal{D}^{-1}_+(w)\mathcal{D}_-(z) +\mathcal{D}^{-1}_-(z)\mathcal{D}_+(w)  \tilde{\phi}^{-1}(z)\tilde{\phi}^{-1}(w) (w/z)^N}{z-w} \frac{\tilde{\phi}(w)-1}{2\pi \ic w} \psi_{i}  \left(\frac{1}{w}\right)  \frac{\dd w  \dd z}{2\pi \ic} \\
\end{split}
\end{equation}
Now we deform the contour of integration for variables $w$ and $z$ respectively to the contours $\T_+$ and $\T_-$ respectively, where $\T_+$ is a circle with radius less than one in the domain of analyticity of $\tilde{\phi}$ and $\psi$, and $\T_-$ is a circle with radius more than one in the domain of analyticity of $\tilde{\phi}$ and $\psi$. So we have
\begin{equation}
\begin{split}
\mathcal{I}_N & \simeq \int_{\T_-}  \int_{\T_+} \frac{- \mathcal{D}^{-1}(w)\mathcal{D}(z) +\mathcal{D}^{-1}(z)\mathcal{D}(w)  \tilde{\phi}^{-1}(z)\tilde{\phi}^{-1}(w) (w/z)^N}{z-w} \frac{\tilde{\phi}(w)-1}{2\pi \ic w} \psi_{i}  \left(\frac{1}{w}\right)  \frac{\dd w  \dd z}{2\pi \ic} \\
& \simeq \int_{\T_-}  \int_{\T_+} \frac{- \mathcal{D}^{-1}(w)\mathcal{D}(z) }{z-w} \frac{\tilde{\phi}(w)-1}{(2\pi \ic)^2 w} \psi_{i}  \left(\frac{1}{w}\right)  \dd w  \dd z \\ & = - \int_{\T_-}  \int_{\T_+}  \mathcal{D}^{-1}(w)\mathcal{D}(z) \sum^{\infty}_{k=0} \left( \frac{w}{z} \right)^k  \frac{\tilde{\phi}(w)-1}{(2\pi \ic)^2 wz} \psi_{i}  \left(\frac{1}{w}\right)  \dd w  \dd z \\ 
& = -\sum^{\infty}_{k=0} \left[ \int_{\T_+} \mathcal{D}^{-1}(w) w^k \frac{\tilde{\phi}(w)-1}{2\pi \ic w} \psi_{i}  \left(\frac{1}{w}\right)  \dd w \right] \left[ \int_{\T_-}    \mathcal{D}(z)  z^{-k}   \frac{\dd z}{2 \pi \ic z} \right]  
\end{split}
\end{equation}
Note that 
\begin{equation}
\begin{split}
\int_{\T_-}    \mathcal{D}(z)  z^{-k}   \frac{\dd z}{2 \pi \ic z} & = \int_{\T}    \mathcal{D}_-(z)  z^{-k}   \frac{\dd z}{2 \pi \ic z} = \int_{\T}    \mathcal{D}_-(z^{-1})  z^{k}   \frac{\dd z}{2 \pi \ic z}  = \int_{\T}    \tilde{\mathcal{D}}_+(z)  z^{k}   \frac{\dd z}{2 \pi \ic z} \\ & = \int_{\T_+}    \tilde{\mathcal{D}}(z)  z^{k}   \frac{\dd z}{2 \pi \ic z}  =  \int_{\T_+}   \left( \tilde{\mathcal{D}}(0) + \sum^{\infty}_{j=1} c_jz^j \right)  z^{k}   \frac{\dd z}{2 \pi \ic z}  \\ & =\tilde{\mathcal{D}}(0)  \int_{\T_+}      z^{k}   \frac{\dd z}{2 \pi \ic z} = \tilde{\mathcal{D}}(0) \de_{k0} = \de_{k0}.
\end{split}
\end{equation}
Thus,
\begin{equation}
\begin{split}
\mathcal{I}_N  & \simeq - \int_{\T_+} \mathcal{D}^{-1}(w)  (\tilde{\phi}(w)-1) \psi_{i}  \left(\frac{1}{w}\right) \frac{\dd w}{2\pi \ic w} \\ & = -  \int_{\T} \mathcal{D}_+^{-1}(w)  (\tilde{\phi}(w)-1) \psi_{i}  \left(\frac{1}{w}\right) \frac{\dd w}{2\pi \ic w} \\ & = -  \int_{\T}   \left(\mathcal{D}_-^{-1}(w)-\mathcal{D}_+^{-1}(w)\right) \psi_{i}  \left(\frac{1}{w}\right) \frac{\dd w}{2\pi \ic w}
\end{split}
\end{equation}
Note that
\[ \int_{\T}   \mathcal{D}_-^{-1}(w) \psi_{i}  \left(\frac{1}{w}\right) \frac{\dd w}{2\pi \ic w} = \int_{\T}   \mathcal{D}_-^{-1}(w^{-1}) \psi_{i}  \left(w\right) \frac{\dd w}{2\pi \ic w} = \]
\[ \int_{\T}   \frac{\al_+(w)}{\al(0)} \psi_{i}  \left(w\right) \frac{\dd w}{2\pi \ic w} = \int_{\T_+}   \frac{\al(w)}{\al(0)} \psi_{i}  \left(w\right) \frac{\dd w}{2\pi \ic w} = \psi_0. \]
Therefore 

\begin{equation}
\begin{split}
\mathcal{I}_N  & \simeq -\psi_0 + \int_{\T}   \mathcal{D}_+^{-1}(w) \psi_{i}  \left(\frac{1}{w}\right) \frac{\dd w}{2\pi \ic w} = -\psi_0 + \int_{\T}   \mathcal{D}_+^{-1}(w) \left( \psi  \left(\frac{1}{w}\right) - \psi_{o}  \left(\frac{1}{w}\right) \right) \frac{\dd w}{2\pi \ic w} 
\end{split}
\end{equation}
Note that $\psi_{o}  \left(1/w\right)$ is an analytic function inside the unit circle with $\psi_{o}  \left(1/w\right) = O(w)$ as $w \to 0$, and thus \[ \int_{\T}   \mathcal{D}_+^{-1}(w)  \psi_{o}  \left(\frac{1}{w}\right) \frac{\dd w}{2\pi \ic w} = 0.  \]
Hence, using this and \eqref{D jump condition} we have
\begin{equation}
\begin{split}
\mathcal{I}_N  & \simeq -\psi_0  + \int_{\T}      \frac{\tilde{\psi}(w)}{\mathcal{D}_-(w) \tilde{\phi}(w)}    \frac{\dd w}{2\pi \ic w} = -\psi_0  + \int_{\T}      \frac{\psi(w)}{\mathcal{D}_-(w^{-1}) \phi(w)}    \frac{\dd w}{2\pi \ic w}
\end{split}
\end{equation}
Note that \[ \mathcal{D}_-(w^{-1}) = \tilde{\mathcal{D}}_+(w) = \frac{\al(0)}{\al_+(w)}, \qandq  \mathcal{D}_+(w^{-1}) = \tilde{\mathcal{D}}_-(w) = \frac{\al(0)}{\al_-(w)}. \]
Therefore
\begin{equation}
\begin{split}
\mathcal{I}_N  & \simeq  -\psi_0  + \int_{\T}      \frac{\al_+(w)\psi(w)}{\al(0) \phi(w)}    \frac{\dd w}{2\pi \ic w} = -\psi_0  + \frac{1}{\al(0)}\int_{\T}      \al_-(w)\psi(w)    \frac{\dd w}{2\pi \ic w}
\end{split}
\end{equation}
Comparing the Wiener-Hopf factorization $\phi(w)=\phi_-(w)\phi_+(w)$ with the scalar Riemann-Hilbert jump condition $\al_+(w)=\al_-(w)\phi(w)$, we can identify $\al_-$ with $\phi^{-1}_-$ and $\al_+$ with $\phi_+$, and thus

\begin{equation}
\mathcal{I}_N   \simeq -\psi_0 + \frac{[\phi^{-1}_-\psi]_0}{[\phi_+]_0}.
\end{equation}
Finally recalling \eqref{F_N phi psi}, and taking the limit $N \to \infty$ we arrive at the conclusion of proposition \ref{p.asymp-1}:

\begin{equation}
F[\phi;\psi] =  \frac{[\phi^{-1}_-\psi]_0}{[\phi_+]_0}. 
\end{equation}

\subsection{Derivation of the symbol pair corresponding to the  next-to-diagonal Ising correlations}\label{Simplifying AuYangPerk}

As it is shown in \cite{YP}, the next-to-diagonal two point correlation function is given
by the following bordered Toeplitz determinant,
\begin{equation}\label{corr2}
\langle \sigma_{0,0}\sigma_{N-1,N} \rangle =
\det \begin{pmatrix}
A_0& \cdots & A_{N-2} & B_{N-1}\\
A_{-1}& \cdots  & A_{N-3}&B_{N-2}  \\
\vdots & \vdots & \vdots & \vdots\\
A_{1-N} &  \cdots  & A_{-1}&B_{0}
\end{pmatrix} , \ \ \ \ \ N > 1,
\end{equation}
where in the notations of \cite{YP},
\begin{equation}\label{a_n-Perk}
A_n = \frac{1}{2\pi}\int_{-\pi}^{\pi}e^{-in\theta}\Phi(\theta) d\theta,
\end{equation}
\begin{equation}\label{Phi-Perk}
\Phi(\theta) = \frac{S-S'^* e^{-i\theta}}{\sqrt{\Omega(\theta)}},
\end{equation}
\begin{equation}\label{Omega-Perk}
\Omega(\theta) = S^2+(S'^*)^2-2SS'^* \cos(\theta),
\end{equation}
\begin{equation}\label{b_n-Perk}
B_n=\frac{1}{2\pi} \int^{\pi}_{-\pi} e^{-in\theta}\Psi(\theta) d \theta,
\end{equation}
and
\begin{equation}\label{Psi-Perk}
\Psi(\theta) = \frac{1}{\sqrt{\Omega(\theta)}} \left( SC'^* - \frac{CS'^*(SS'^* +e^{-i\theta})}{CC'^*+\sqrt{\Omega(\theta)}} \right).
\end{equation}
The quantities  $C,S,C'^*,$ and $S'^*$ are determined by the physical parameters of the model
according to the equations,
\begin{equation}\label{Perk's C,S's}
C:= \cosh(2K), \quad S:= \sinh(2K), \quad C'^*:= \frac{\cosh(2K')}{\sinh(2K')}, \quad S'^*:= \frac{1}{\sinh(2K')},
\end{equation}
where
\begin{equation}\label{K,K'-->J_h,J_v}
K = \frac{J_h}{k_BT},  \qquad \mbox{and} \qquad K' = \frac{J_v}{k_BT}.    
\end{equation}
Using \eqref{SSCC} and \eqref{k-parameter} we can write \eqref{Phi-Perk}, \eqref{Omega-Perk} and \eqref{Psi-Perk} in our notations as:
\begin{equation}\label{Omega-Perk1}
\Omega(\theta) =  \frac{k^2+1-2k\cos(\theta)}{S^2_v},
\end{equation}
\begin{equation}\label{Phi-Perk1}
\Phi(\theta) = \frac{k- e^{-i\theta}}{\sqrt{k^2+1-2k\cos(\theta)}},
\end{equation}
and
\begin{equation}\label{Psi-Perk1}
\Psi(\theta) = \frac{1}{\sqrt{k^2+1-2k\cos(\theta)}} \left( S_hC_v - \frac{C_h(S_h +S_ve^{-i\theta})}{C_hC_v+\sqrt{k^2+1-2k\cos(\theta)}} \right).
\end{equation}
Recall that $\widehat{\phi}$ is given by \eqref{Isingphi} is
\begin{equation}\label{*}
\widehat{\phi}(z)= \sqrt{\frac{1-k^{-1}z^{-1}}{1-k^{-1}z}}=\frac{k-z^{-1}}{\sqrt{k^2+1-k(z+z^{-1})}}.
\end{equation}
This together with \eqref{Phi-Perk1} immediately yields 
\begin{equation}
\Phi(\theta) = \widehat{\phi}(e^{i\theta}).
\end{equation}
Next we want to show that $\Psi(\theta)=\widehat{\psi}(e^{i\theta})$. To that end note that
\begin{equation}\label{I_2integrand1}
\frac{1}{C_hC_v+\sqrt{k^2+1-k(z+z^{-1})}} = \frac{C_hC_v - \sqrt{k^2+1-k(z+z^{-1})}}{S_h^2+S_v^2+k(z+z^{-1})} = \frac{C_hC_vz - z\sqrt{k^2+1-k(z+z^{-1})}}{k(z-c_*)(z-c^{-1}_*)},
\end{equation}
where
\begin{equation}\label{c}
c_* := - \frac{S_h}{S_v}.
\end{equation}
Therefore, as $S_h+S_vz^{-1}=S_hz^{-1}(z-c_*^{-1})$,
\begin{equation}
\frac{S_h+S_v z^{-1}}{C_hC_v+\sqrt{k^2+1-k(z+z^{-1})}} =  \frac{S_hC_hC_v }{k(z-c_*)}  - \frac{S_h  \sqrt{k^2+1-k(z+z^{-1}) }}{k(z-c_*)}. 
\end{equation}

Combining this with \eqref{Psi-Perk1} gives
\begin{align*}
\Psi(\theta) &= \frac{S_hC_v }{\sqrt{k^2+1-k(z+z^{-1})}} -\frac{S_hC^2_hC_v }{k(z-c_*)\sqrt{k^2+1-k(z+z^{-1})}}+\frac{S_hC_h }{k(z-c_*)}
\\
&=
\frac{S_hC_v}{\sqrt{k^2+1-k(z+z^{-1})}} \left[ 1-\frac{C_h^2}{k(z-c_*)}\right]+\frac{S_hC_h }{k(z-c_*)},
\qquad\qquad z=e^{i\theta}.
\end{align*}
and the term in the brackets becomes $(kz-1)/(k(z-c_*))$. Now, using \eqref{*} 
we obtain the formula for $\widehat{\psi}$ given by \eqref{hat psi}:
\begin{equation}
\Psi(\theta) = \frac{C_v z\widehat{\phi}(z)+C_h}{S_v(z-c_*)} \equiv \widehat{\psi}(z). 
\end{equation}
Let us also remark that 
\begin{equation}
\widehat{\psi}(z)= \frac{C_v}{S_v}\cdot \frac{z\widehat{\phi}(z)-c_*\widehat{\phi}(c_*)}{z-c_*},
\end{equation}
which can be seen from a straightforward computations as well. To summarize, we have shown that 
\begin{equation}
\langle \sigma_{0,0}\sigma_{N-1,N} \rangle = D_N[\widehat{\phi};\widehat{\psi}].
\end{equation}
with $\widehat{\phi}$ and $\widehat{\psi}$ given by  \eqref{Isingphi} and  \eqref{hat psi}.


\subsection{Auxiliary results}

\begin{lemma}\label{l.sum.asym}
	Let $|a|<1$ and $\omega$ be complex parameters. Then
	$$
	\sum_{n=N}^\infty a^n n^\omega(1+O(n^{-1}))=\frac{a^N N^\omega}{1-a}(1+O(N^{-1})),\qquad  N\to\infty.
	$$
\end{lemma}
\begin{proof}
	We basically can apply summation by parts,
	\begin{align*}
	(1-a)\sum_{n=N}^\infty a^n n^\omega(1+O(n^{-1})) 
	&=
	\sum_{n=N}^\infty a^n n^\omega(1+O(n^{-1})) -
	\sum_{n=N}^\infty a^{n+1} n^\omega(1+O(n^{-1})) 
	\\
	&=a^NN^\omega(1+O(N^{-1}))+
	\\
	&\qquad \sum_{n=N}^\infty
	a^{n+1}\Big((n+1)^\omega(1+O(n^{-1}))-n^\omega(1+O(n^{-1}))\Big)
	\\
	&=a^NN^\omega(1+O(N^{-1}))+
	\sum_{n=N}^\infty
	a^{n+1}O(n^{\omega-1}).
	\end{align*}
	The last term we can split into
	$$
	\sum_{n=N}^{2N-1}
	a^{n}O(n^{\omega-1})=
	O(a^NN^{\omega-1}),\qquad
	\sum_{n=2N}^{\infty}
	a^{n}O(n^{\omega-1})=
	\sum_{n=2N}^\infty
	O(a^nq^n)=O((aq)^{2N}).
	$$
	In the latter we choose $1 < q< |a|^{-1}$, which guarantees that 
	$n^{\omega-1}=O(q^n)$ and $(aq)^{2N}=O(N^{\omega-1})$.
\end{proof}

\begin{lemma}\label{l.sum.asym2}
Let $|a|<1$ and $\omega$ be complex parameters. Then
	$$
	\sum_{n=0}^\infty (n+1)(n+N)^\omega a^{n+N}=\frac{a^N N^\omega}{(1-a)^2}(1+O(N^{-1})),\qquad
	N\to\infty.
	$$
\end{lemma}
\begin{proof}
	After dividing by $a^N$, the difference between the series and the leading term is
	\begin{align*}
	\lefteqn{
		\sum_{n=0}^\infty (n+1)(n+N)^\omega a^{n}-\frac{N^\omega}{(1-a)^2}
		=
		\sum_{n=0}^\infty (n+1) a^{n}\Big((n+N)^\omega -N^\omega\Big)
	}\hspace{.5cm}
	\\
	&=
	\sum_{n=0}^\infty (n+1)^2 a^{n}O(\max\{(n+N)^{\re(\omega)-1},N^{\re(\omega)-1}\})
	=O(N^{\re(\omega)-1}).
	\end{align*}
	This implies the estimate.
\end{proof}

\begin{lemma}\label{l.Fc.asym}
	Let $\zeta(z)$ be a function holomorphic on $\{z\in \C\,:\,1-\eps<|z|<b+\eps\}\setminus[b,b+\eps)$
	with $b>1$, $\eps>0$. Further assume that in some neighborhood of $[b,b+\eps)$ this function is of the form
	$$
	\zeta(z)=(b-z)^{\omega}\xi(z)+\zeta_0(z)
	$$
	with $\xi(z)$ and $\zeta_0(z)$ being holomorphic, and $\re(\omega)>-1$.
	Then the Fourier coefficients of $\zeta$ have the asymptotics
	$$
	\zeta_n = 
	\Big(\frac{\xi(b)}{\Gamma(-\omega)}+O(n^{-1})\Big)n^{-\omega-1}b^{-n+\omega}
	\qquad n\to +\infty.
	$$
\end{lemma}
\begin{proof}
	In the formula for the Fourier coefficients we deform the contour into a slightly bigger circle with radius $b(1+\delta_n)$
	(where $\delta_n=\delta n^{-1/2}$ and $0<\delta<\eps/b$ is fixed)
	and a line segments along the branch cut $[b,b+b\delta_n]$ on both sides,
	\begin{align*}
	\zeta_n &=
	\frac{1}{2\pi i}\int_{|z|=1}\zeta(z)z^{-n-1}\,dz
	\\
	&=\frac{1}{2\pi i}\int_{|z|=b(1+\delta_n)}\zeta(z)z^{-n-1}\,dz
	+
	\frac{1}{2\pi i}\int_{b}^{b(1+\delta_n)} \left((b-t-i0)^\omega-(b-t+i0)^\omega\right)\xi(t)t^{-n-1}\, dt.
	\end{align*}
	The first integral being $O(\delta_n^{-|\re(\omega)|}b^{-n}(1+\delta_n)^{-n})=O(n^{|\re(\omega)|/2}b^{-n}e^{-n^{1/2} \delta})$ is negligible. The second one becomes
	$$
	-\frac{\sin(\omega\pi)}{\pi }\int_{b}^{b(1+\delta_n)} (t-b)^\omega \xi(t)t^{-n-1}\, dt
	=
	-\frac{\sin(\omega\pi)}{\pi } b^{\omega-n} \int_0^{\delta_n} \xi(b+bs)s^\omega (1+s)^{-n-1}\,ds.
	$$
	Therein, the integral (without the factors in front of it) equals
	\begin{align*}
	\lefteqn{
		\int_0^{\delta_n} (\xi(b)+O(s))s^\omega e^{-(n+1)(s+O(s^2))}\,ds
	}
	\qquad\qquad\\
	& =
	n^{-\omega-1}\int_0^{n^{1/2}\delta} \Big(\xi(b)+O({\textstyle \frac{u}{n}})\Big)
	u^\omega e^{-u+O(\frac{u+u^2}{n})}\, du
	\\
	&=
	\xi(b) n^{-\omega-1} \int_0^{n^{1/2}\delta} u^\omega e^{-u}\, du +
	n^{-\omega-2}\int_0^{n^{1/2}\delta}u^\omega O(u+u^2)e^{-u}\, du
	\\
	&=
	n^{-\omega-1}\Big(\xi(b)\Gamma(1+\omega) + O(n^{-1})\Big).
	\end{align*}
	Combining all this give the asymptotic formula.
\end{proof}


\textbf{Acknowledgements.} The authors would like to thank Pavel Bleher, Barry McCoy, Vitaly Tarasov, and Nicholas Witte for their interest in this project and for helpful conversations. EB, TE and RG acknowledge American Institute of Mathematics for providing excellent working conditions and their support during the SQuaRE program "Asymptotic behavior of Toeplitz and Toeplitz+Hankel determinants" where part of their work was done during the 2019 and 2020 meetings. EB acknowledges support from the NSF grant DMS-2050092. TE acknowledges support from the Simons Foundation Collaboration Grant \# 525111. AI acknowledges support from the NSF grant  DMS-1955265. 


\end{document}